\documentclass[11pt,%
letterpaper]{article}
\usepackage{fullpage,amsmath,amssymb,amsthm}
\usepackage{xspace}
\usepackage{ifpdf}
\usepackage{color}

\usepackage[ruled,vlined,linesnumbered]{algorithm2e}

\usepackage{ifpdf}
\usepackage{hyperref}

\newcommand{\Hrzline}{\hrule height.8pt depth0pt\kern2pt}

\def\whp{with high probability\xspace}
\def\tO{\tilde{O}}
\def\tOmega{\tilde{\Omega}}

\newcommand{\REMOVE}[1]{}

\newcommand{\eps}{\varepsilon}

\newcommand{\tuple}[1]{{\langle{#1}\rangle}}
\newcommand{\minn}[1]{\min\{#1\}}

\newcommand{\aset}[1]{\{#1\}}

\def\OPT{\mathrm{OPT}}
\newcommand{\expec}{{\mathbb{E}}}
\newcommand{\prob}{{\rm Pr}}

\def\Schaffer{Sch{\"a}ffer}
\def\LD2S{{\sc ld2s}\xspace}
\def\DkS{{\sc d$k$s}\xspace}
\def\SmES{{\sc s$m$es}\xspace}
\def\smes-lp{\hyperlink{smes-lp}{\textbf{SmES-LP}_q}}
\def\algA{algorithm ${\mathcal A}$\xspace}

\newcommand{\mdnote}[1]{}%
\newcommand{\enote}[1]{}%

\newtheorem{theorem}{Theorem}[section]
\newtheorem{lemma}[theorem]{Lemma}
\newtheorem{prop}[theorem]{Proposition}
\newtheorem{claim}[theorem]{Claim}

\newtheorem{corollary}[theorem]{Corollary}

\newtheorem{hypothesis}[theorem]{Hypothesis}

\theoremstyle{remark}
\newtheorem{remark}[theorem]{Remark}

\theoremstyle{definition}
\newtheorem{definition}[theorem]{Definition}

\def\compactify{\itemsep=0pt \topsep=0pt \partopsep=0pt \parsep=0pt}
\def\Lovasz{Lov\'asz\xspace}
\newcommand{\E}{\mathop{\mathbb E}}
\newcounter{this-list}

\newcounter{par-list}
\newlength{\parlistlength}

\begin{document}

\title{Everywhere-Sparse Spanners via Dense Subgraphs}

\author{
Eden Chlamt\'a\v{c}%
\thanks{Research supported in part
by an ERC Advanced grant.
Email: \texttt{chlamtac@post.tau.ac.il}
}
\\ Tel Aviv University
\and
Michael Dinitz%
\thanks{%
Work supported in part by an Israel Science Foundation grant \#452/08,
a US-Israel BSF grant \#2010418, and by a Minerva grant.
Email: \texttt{\{michael.dinitz,robert.krauthgamer\}@weizmann.ac.il}
}
\\ The Weizmann Institute
\and
Robert Krauthgamer%
\footnotemark[2]
\\ The Weizmann Institute
}

\maketitle

\begin{abstract}
  The significant progress in constructing graph spanners that are sparse (small number of edges) or light (low total weight) has skipped spanners that are everywhere-sparse (small maximum degree).  This disparity is in line with other network design problems, where the maximum-degree objective has been a notorious technical challenge.  Our main result is for the Lowest Degree 2-Spanner (LD2S) problem, where the goal is to compute a 2-spanner of an input graph so as to minimize the maximum degree.  We design a polynomial-time algorithm achieving approximation factor $\tO(\Delta^{3-2\sqrt{2}}) \approx \tO(\Delta^{0.172})$, where $\Delta$ is the maximum degree of the input graph.  The previous  $\tO(\Delta^{1/4})$--approximation was proved nearly two decades ago by Kortsarz and Peleg [SODA 1994, SICOMP 1998].

  Our main conceptual contribution is to establish a formal connection between LD2S and a variant of the Densest k-Subgraph (DkS) problem.  Specifically, we design for both problems strong relaxations based on the Sherali-Adams linear programming (LP) hierarchy, and show that ``faithful'' randomized rounding of the DkS-variant can be used to round LD2S solutions.  Our notion of faithfulness intuitively means that all vertices and edges are chosen with probability proportional to their LP value, but the precise formulation is more subtle.

  Unfortunately, the best algorithms known for DkS use the Lov\'asz-Schrijver LP hierarchy in a non-faithful way [Bhaskara, Charikar, Chlamtac, Feige, and Vijayaraghavan, STOC 2010].  Our main technical contribution is to overcome this shortcoming, while still matching the gap that arises in random graphs by planting a subgraph with same log-density.
\end{abstract}

\section{Introduction}

The significant progress made over the years in constructing graph spanners
shares, for the most part, two features: (1) the objective
is to minimize the total number/weight of edges; and (2) the techniques are
primarily combinatorial.  This second feature has started to change
recently, with the use of Linear Programming (LP) in several results~\cite{BGJRW09,BRR10,DK11a,BBMRY11}.  One of the earliest uses of
linear programming for spanners, though, was also one of the few
examples of a different objective function: in 1994, Kortsarz and Peleg~\cite{KP98}
considered the {\sc Lowest Degree $2$-Spanner} (\LD2S) problem,
where the goal is to find a $2$-spanner of an input graph that minimizes the maximum degree,
and used a natural LP relaxation to devise a polynomial-time algorithm
achieving approximation factor $\tO(\Delta^{1/4})$
(%
where $\Delta$ is the maximum degree).
They also showed that it is NP-hard to approximate \LD2S within a factor
smaller than $\Omega(\log n)$.
We make the first progress on approximating \LD2S since then,
by designing a new approximation algorithm with an improved approximation factor.

\begin{theorem} \label{thm:main}
For an arbitrarily small fixed $\eps>0$,
the \LD2S problem can be approximated in polynomial time within factor
$\tO(\Delta^{3-2\sqrt{2}+\eps}) \leq \tO(\Delta^{0.172})$.
\end{theorem}

Degree bounds have a natural mathematical appeal and are also useful in many applications.
For example, one common use of spanners is in compact routing schemes (e.g.~\cite{TZ01r,Din07}), which store small routing tables at every node.  If we route on a spanner with large maximum degree, then \emph{a priori} the node of large degree will have a large table, even if the total number of edges is small.  Similarly, the maximum degree (rather than the overall number of edges) is what determines local memory constraints when using spanners to construct network synchronizers~\cite{PU89} or for efficient broadcast~\cite{ABP92}.  The literature on approximation algorithms includes recent exciting work on sophisticated LP rounding for network design problems involving degree bounds (e.g.~\cite{LNSS09,SL07}).

\paragraph{Dense subgraphs.}
Our central insight involves the relationship between
\emph{sparse} spanners and finding \emph{dense} subgraphs.
Such an informal relationship has been
folklore in the distributed computing and approximation algorithms communities;
for instance, graph spanners are mentioned as the original motivation for
introducing the {\sc Densest $k$-Subgraph} (\DkS) problem~\cite{KP93}, %
in that case in the context of minimizing the total number of edges in the spanner.
Surprisingly, we show that there is a natural connection between \DkS and the more challenging task of constructing spanners that have small maximum degree.  We prove that certain types of ``faithful'' approximation algorithms for a variant of \DkS which we call {\sc Smallest $m$-Edge Subgraph} (or \SmES) imply approximation algorithms for \LD2S,
and then show how to construct such an algorithm for \SmES;
combining these two together yields our improved approximation for \LD2S.

We seem to be the first to formally define and study \SmES, although it has been used in previous work (sometimes implicitly) as the natural minimization version of \DkS, see e.g.~\cite{Nut10,GHNR07,AGGN10}.
A straightforward argument shows that an $f$-approximation for \SmES implies an $\tO(f^2)$-approximation for \DkS. In the other direction, all that was known was that an $f$-approximation for \DkS implies an $\tO(f)$-approximation for \SmES.  One contribution of this paper is a non-black box improvement: while the best-known approximation for \DkS is $O(n^{1/4 + \eps})$, we give an $O(n^{3-2\sqrt{2} + \eps})$-approximation for \SmES.
This improvement is key to our main result about approximating \LD2S.

\paragraph{LP hierarchies.}
The log-density framework introduced in~\cite{BCCFV10} in the context of \DkS
(see Section~\ref{sec:DkS-overview}) predicts, when applied to \SmES,
that current techniques would hit a barrier at $n^{3-2\sqrt{2}}$,
precisely the factor achieved by our algorithm.
Here, the use of strong relaxations (namely LP hierarchies) is crucial, since simple relaxations have large integrality gaps. For example, one can show that the natural SDP relaxation for \SmES has an $\Omega(n^{1/4})$ integrality gap (for $G=G(n,n^{-1/2})$ and $m=n^{1/2}$), similarly to the $\Omega(n^{1/3})$-gap shown for \DkS by Feige and Seltser~\cite{FS97}.

While we borrow some of the algorithmic techniques developed for \DkS by \cite{BCCFV10}, the crucial need for a ``faithful'' approximation required us to develop new tools which represent a significant departure from previous work both in terms of the algorithm and its analysis. %
For example, our algorithm and analysis rely on the existence of consistent high-moment variables arising from the Sherali-Adams~\cite{SA90} hierarchy (see, e.g.\ Lemma~\ref{lem:regularity}) and not present in the \Lovasz-Schrijver~\cite{LS91} LP hierarchy (which was sufficient for~\cite{BCCFV10}).

\paragraph{Basic terminology.}
We denote the (undirected)\footnote{Our algorithm for \LD2S also works for the directed case, though for simplicity we focus on undirected graphs.} input graph by $G=(V,E)$, and let $n=|V|$.
For a vertex $v\in V$, let $\Gamma_G(v)=\aset{u:\ \aset{u,v}\in E}$
denote its neighbors in $G$.  If the graph $G$ is clear from context then we will drop the subscript and simply refer to $\Gamma(v)$.
Recall that the maximum degree of vertices in $G$ is denoted $\Delta$.
We suppress polylogarithmic factors by using the notation $\tO(f)$
as a shorthand for $f\cdot (\log n)^{O(1)}$.

As usual, a \emph{$2$-spanner} of $G$ is a subgraph $H=(V,E_H)$
such that every $u,v\in V$ that are connected by an edge in $G$
are also connected in $H$ by a path of length at most $2$.  This is a special case of the more general notion of a \emph{$k$-spanner}, which was introduced by Peleg and \Schaffer~\cite{PS89}
and has been studied extensively;
see also Section~\ref{sec:related}.

\subsection{LP-based approach for \LD2S}

The LP relaxation of \LD2S used by Kortsarz and Peleg \cite{KP98} is
very natural: for each edge $\{u,v\} \in E$ it has a variable $x_{\{u,v\}}\in[0,1]$,
plus additional variables $x_{\{u,v\};w}\in[0,1]$ for every $w\in\Gamma(u)\cap\Gamma(v)$
(i.e., whenever $u,v,w$ form a triangle in $G$).
The objective is to minimize $\lambda$, subject to a degree constraint
\begin{align}
  \textstyle \label{eq:lp-KP-first}
  \sum_{v \in\Gamma(u)} x_{\{u,v\}} \leq \lambda
  \qquad\qquad & \forall u\in V,
\intertext{and the constraints that every edge in $G$ (i.e.~demand pair)
is covered by either a $1$-path or a $2$-path in the spanner (subgraph):}
  \textstyle %
  x_{\{u,v\}} + \sum_{w \in \Gamma(u) \cap \Gamma(v)} x_{\{u,v\};w} \geq 1
  \qquad & \forall \{u,v\} \in E.
\\ \textstyle \label{eq:lp-KP-last}
  x_{\{u,v\};w} \leq \minn{x_{\{u,w\}},x_{\{v,w\}}}
  \qquad & \forall \{u,v\} \in E,\ w \in \Gamma(u) \cap \Gamma(v).
\end{align}

This LP relaxation seems like a natural place to start,
but it is actually quite weak, having integrality gap $\Omega(\sqrt{\Delta})$.
Indeed, let $G$ be a clique of size $\Delta+1$;
observe that every $2$-spanner of this $G$ must have
maximum degree at least $\sqrt{\Delta}$,
while the LP has value $\lambda \leq 1$
(by setting all $x$ variables to $1/\Delta$).
The same argument works for a disjoint union of $n/(\Delta+1)$ such cliques.  Kortsarz and Peleg \cite{KP98} nevertheless managed to achieve
$\tO(\Delta^{1/4})$ approximation (in polynomial-time).
Their algorithm combines a relatively simple rounding of this LP
with another partial solution that does not use the LP,
and whose analysis relies on a combinatorial lower bound on the optimum.

Our approach is to look at the Kortsarz-Peleg LP above from the perspective of a single vertex $w$.  Consider an integral solution $H$ to the LP above, i.e.~a valid $2$-spanner.
From the viewpoint of $w$, incident edges are included in $H$ for two possible reasons: either to span an edge connecting two neighbors of $w$ (i.e., including the edges $\{u,w\}$ and $\{v,w\}$ in order to span the edge $\{u,v\}$), or to span the edge itself.
It's reasonable to focus on the case where $H$ has significantly fewer edges than $G$,
and therefore many edges in $H$ are included because of the first reason.
Let $G_w$ be the subgraph of $G$ induced by the neighbors of $w$, and let $S$ be the subset of vertices of $G_w$ that are adjacent to $w$ in $H$.
Then from the perspective of $w$, including the edges between $w$ and $S$ in $H$ ``covers'' every demand formed by an edge (of $G_w$) that connects two vertices in $S$,
namely $E'=\aset{\aset{u,v}\in E:\ u,v\in S}$.
We can look at each neighborhood this way, and reinterpret \LD2S
as the problem of covering every demand in at least one neighborhood $G_w$,
while minimizing the maximum degree.

This viewpoint naturally suggests an LP-based algorithm for LD2S: solve the Kortsarz-Peleg LP above (or some other relaxation),
and for every $w\in V$, interpret $\sum_{\{u,v\} \in E : u,v \in \Gamma(w)} x_{\{u,v\};w}$ as the amount of ``demand'' that $w$ is supposed to cover locally, and $\sum_{u \in \Gamma(w)} x_{\{w,u\}}$ as $w$'s ``budget''.  Then for each $w \in V$ run a subroutine that covers the required amount of demand within the budget.  Since in \LD2S every demand \emph{must} be covered, this subroutine should cover the required amount of demand but is free to somewhat violate the budget constraint; the amount of violation will correspond to the \LD2S approximation guarantee.
We thus need to solve the {\sc Smallest $m$-Edge Subgraph} (\SmES) problem:
given a graph (in our case $G_w$) and a value $m$, choose as few vertices as possible subject to covering at least $m$ edges, where an edge is covered if both its endpoints are chosen.
Unfortunately, this reduction from \LD2S to \SmES does not work.  There are two main issues with it.
First, if $w$ chooses to add an edge to $u$ (i.e.~the \SmES algorithm at $G_w$ includes $u\in\Gamma(w)$) then this increases the degree of \emph{both} $w$ and $u$.
So even if $u$ stays within its own budget when $G_u$ is processed, many of its neighbors might decide to add their edge to $u$, and the degree at $u$ will be very large compared to its budget.
Second, since we run a \SmES algorithm at each vertex \emph{separately}, they might make poorly-correlated choices as to which demands they cover. This may cause a high degree of overlap in the demands covered by different vertices, leading to much less total demand covered. %
 Both of these problems stem from the same source: while we used the LP to define the total demand and budget at each vertex, we did not require the \SmES algorithm to act in a way consistent with the LP. %
 If we could force the \SmES subroutine to make decisions that actually correspond to the fractional solution, then both of these problems would be solved. This is our motivation for defining faithfulness.

\subsection{Faithful rounding}

While our formal notion of faithfulness is somewhat technical and depends on the exact problem that we want to solve, the intuition behind it is natural and can apply to many problems.  Suppose that we have an LP in which there are variables $\{x_e\}_{e \in U}$ (where $U$ is a universe of elements) as well as variables $\{x_{e, e'}\}_{e, e' \in U}$.  In our case, each $e$ is a vertex in a \SmES instance (i.e.~an edge in \LD2S) and each pair $\{e, e'\}$ is an edge in a \SmES instance (a $2$-path in \LD2S).  A standard way of interpreting fractional LP values is as \emph{probabilities}, i.e.~we think of $x_e$ as the probability that $e$ should be in the solution.
This interpretation naturally leads to independent randomized rounding, where we take $e$ into our solution with probability proportional to $x_e$.
By this interpretation, $x_{e, e'}$ should be the probability that both $e$ and $e'$ are in the solution.  But now we have a problem, since the natural constraints to force this type of situation in an integral setting,
namely constraints such as $x_{e,e'} \leq \minn{x_e,x_{e'}}$, correspond poorly to the probabilities obtained by independent randomized rounding.
For example, if $x_e = x_{e'} = x_{e, e'}$, then the LP ``believes'' that the probability that both $e$ and $e'$ are in the solution is $x_{e,e'}$, but under independent randomized rounding this event happens with probability $x_e\cdot x_{e'}=x_{e,e'}^2$, which could be much smaller. In a faithful rounding this does not happen: roughly speaking, faithfulness requires every element and pair of elements to be included in the solution with probability that is proportional to its LP value.

Many algorithms are naturally faithful, and indeed we suspect that one reason this notion has not been defined previously (to the best of our knowledge) is that in most cases it either falls out from the analysis ``for free'' or it is unnecessary.
The connection we show between \LD2S and faithful rounding for \SmES
might give one hope that the recent algorithmic breakthrough for \DkS by
Bhaskara, Charikar, Chlamtac, Feige and Vijayaraghavan~\cite{BCCFV10} could
imply better approximations for \LD2S. %
However, their result heavily uses hierarchies, which creates a formidable obstacle for faithful rounding, as we discuss in Section~\ref{sec:hierarchies}.

\subsection{LP hierarchies and faithful rounding}\label{sec:hierarchies}

Following the lead of Bhaskara et al.~\cite{BCCFV10},
we employ a strong LP relaxation for \SmES, which can be viewed as part of an LP hierarchy. %
In this context, a hierarchy is a sequence of increasingly %
tight relaxations to a 0-1 program,
usually obtained via a general mechanism that works for any 0-1 program. Such hierarchies (for both LPs and SDPs) have been suggested by Sherali and Adams~\cite{SA90}, \Lovasz and Schrijver~\cite{LS91}, and Lasserre~\cite{Lasserre02} (in our case, we use the Sherali-Adams hierarchy). A key property shared by these hierarchies is that they are locally integral;
that is, the $q$-th relaxation in the hierarchy coincides exactly with the convex hull of feasible 0-1 solutions, when both are projected onto any $q$-dimensional subspace corresponding to $q$ variables in the program.\footnote{Consequently,
if $N$ denotes the number of initial 0-1 variables, 
then the $N$-th relaxation is exactly the convex hull of all 0-1 solutions,
i.e., corresponds to solving the 0-1 optimization problem exactly.
The $q$-th relaxation in the sequence
can be written explicitly as a (linear) program of size $N^{O(q)}$,
and thus solved in time $N^{O(q)}$.
} %
Specifically for Sherali-Adams, the $q$-th relaxation for a given 
0-1 linear program
with variables $x_1,\ldots,x_N\in \aset{0,1}$, is obtained by extending the
0-1 program to include a variable $x_S$ for every $S\subseteq\aset{1,\ldots,N}$, $|S|\leq q$,
and then writing a ``locally integral'' relaxation for this extended 0-1 program
to guarantee that $x_S=\prod_{i\in S} x_i$ (by convention $x_\emptyset=1$).
For more details, see the survey \cite{CT12}.

There has been a recent surge of interest in the study of hierarchies of LPs (or other convex programs),
especially in connection with approximation algorithms for
combinatorial optimization problems.
Specifically, such strong relaxations can potentially lead to
progress on problems whose approximability has persistent gaps,
such as {\sc Vertex-Cover} and {\sc Minimum-Bisection}.
This line of attack was probably first described explicitly in \cite{ABL02}. %
However,
designing rounding procedures for these relaxations is often quite challenging.
Indeed, relatively few papers have managed to improve over
state-of-the-art approximation algorithms using hierarchies.  The few papers that do give improved approximation bounds using hierarchies
include~\cite{Chlamtac07, CS08, BCG09, CKR10, BCCFV10}.%
\footnote{There are also papers that recover known approximation bounds, say a PTAS,
while other ones show the limitations of these hierarchies 
by exhibiting integrality gaps for certain problems and hierarchies.
}
In particular, the last paper
designs a rounding procedure for an LP hierarchy for \DkS, %
which we adapt for \SmES.

Our plan is to leverage the success of \cite{BCCFV10},
but as mentioned before,
we face a serious obstacle --- their rounding procedure is not faithful.
They essentially condition on a small set of events,
for instance that the solution includes
a small set $S^*$ of \emph{carefully chosen} elements,
and then they use only the LP variables for sets containing this $S^*$,
namely, a variable $x_{S^*\cup \{u\} }$ is now thought of
as the LP variable for singleton $u$.
But clearly that variable might have very little to do with the actual $x_u$,
which is the quantity with respect to which we are trying to be faithful.

Our main technical contribution is to overcome this and design
a faithful rounding for \SmES based on Sherali-Adams.
Our algorithm is loosely based on the \DkS algorithm of~\cite{BCCFV10}, but numerous technical difficulties have to be resolved to make it faithful.  This, together with our reduction from \LD2S to faithful \SmES, gives our new approximation algorithm for \LD2S.  We believe that our notion of faithful rounding is of independent interest, and might prove useful for other approximation algorithms, especially in the context of using hierarchies such as Sherali-Adams.

For comparison, we mention that recent algorithmic results,
due to \cite{BRS11,GS11},
design rounding schemes for the Lasserre~\cite{Lasserre02} hierarchy.
Their rounding %
 appears to be faithful
(at least at an informal level), but it is not applicable to our context.
First, their analysis holds only for expander-like graphs,
and second, their rounding technique applies to problems such
as constraint satisfaction and graph partitioning, with no connection to \DkS.

\subsection{Related work}
\label{sec:related}

Graph spanners, first introduced by Peleg and Sch{\"a}ffer~\cite{PS89} and Peleg and Ullman~\cite{PU89}, have been studied extensively, with applications ranging from routing in networks (e.g. \cite{AP95,TZ05}) to solving linear systems (e.g. \cite{ST04a,EEST08}).  The foundational result on spanners is due to Alth\"ofer, Das, Dobkin, Joseph and Soares \cite{ADDJS93}, who gave an algorithm that, given a graph and an integer $k \geq 1$, constructs a $(2k-1)$-spanner with $n^{1+1/k}$ edges.  Unfortunately this result obviously does not give anything nontrivial for $2$-spanners, and indeed it is easy to see that there exist graphs for which every $2$-spanner has $\Omega(n^2)$ edges,
thus nontrivial absolute bounds on the size of a $2$-spanner are not possible.
Kortsarz and Peleg~\cite{KP94} were the first to consider relative bounds
for spanners. They gave a greedy $O(\log |E|/|V|)$-approximation algorithm for the problem of finding a $2$-spanner with the minimum number of edges.  This was then extended to variants of $2$-spanners, e.g.~\emph{client-server $2$-spanner}~\cite{EP01} and \emph{fault-tolerant $2$-spanner}~\cite{DK11a,DK11b} (for which only $O(\log \Delta)$ is known).  All of these bounds are basically optimal,
assuming $\mathrm{P}\neq\mathrm{NP}$,
due to a hardness result of Kortsarz~\cite{Kortsarz01}.

\subsection{Outline}
We begin in Section~\ref{sec:overview} by giving a high-level overview of our reduction from \LD2S to \SmES and our faithful rounding for \SmES. This overview is not technically accurate, but provides a simplified algorithm and analysis in order to provide the intuition behind our approach. In Section~\ref{sec:LPs} we give a formal description of the LP relaxations for both problems. In Section~\ref{sec:reduction} we give the details behind our reduction. Section~\ref{sec:SmES-algorithm} contains preliminaries for our \SmES rounding algorithm, while the algorithm itself can be found in Section~\ref{sec:SmES-details}, and the analysis is in Section~\ref{sec:SmES-analysis}. Finally, we conclude with a discussion of future directions Section~\ref{sec:discussion}.

\section{Overview of our approach}\label{sec:overview}

\subsection{Overview of LP relaxation for \LD2S and reduction to \SmES}

In this section we give an LP relaxation for \LD2S that uses a relaxation of \SmES as a black box, as well as an algorithm that shows how to use a faithful rounding for \SmES to approximate \LD2S.  Both the relaxation and the algorithm presented here are simplifications that ignore some technical details; the full relaxation and algorithm, as well as all proofs, can be found in Sections~\ref{sec:LPs} and~\ref{sec:reduction}.

We will actually give a relaxation for a slightly more general version of \LD2S in which instead of spanning \emph{all} edges we are given a subset $\hat E \subseteq E$ and are only required to span edges in $\hat E$.  Note that the optimal solution for demands $\hat E \subseteq E$ has maximum degree that is at most the maximum degree of the optimal solution to the original \LD2S problem (where all edges are demands).  This will allow us to cover some demands, re-solve the LP with only the remaining demands, and repeat.

Our relaxation is a feasibility LP, so we will guess the optimal degree bound $\lambda$ and use it as a constant in the LP.
For each $u \in V$, let $G_u=(V_u,E_u)$ be the induced subgraph of $(V, \hat E)$ on $\Gamma_G(u)$.
Our relaxation includes a fractional \SmES solution for each $G_u$: let $\textbf{SmES-LP}(G_u)$ be a linear relaxation of \SmES with variables $\{z^u_v\}_{v \in V_u} \cup \{z^u_e\}_{e \in E_u}$ with the property that $z^u_{\{w,v\}} \leq \min\{z^u_w, z^u_v\}$ for all $\{w,v\} \in E_u$.  In a 0-1 solution this means an edge is covered only if both of its endpoints are chosen.  Any polytope that includes this basic condition can be used, but obviously the tighter this relaxation is the tighter our \LD2S relaxation will be, and in the end we will use a much stronger relaxation for \SmES that is based on the Sherali-Adams hierarchy.

Our relaxation for \LD2S with demands $\hat E \subseteq E$ is given by~\eqref{eq':LD2S_LP_SmES}-\eqref{eq':LD2S_LP_cover}.

\begin{align}
\label{eq':LD2S_LP_SmES} &((z^{u}_v)_{v \in V(G_u)}, (z^{u}_e)_{e \in E(G_u)}) \in \textbf{SmES-LP}(G_u) & \forall u \in V\\
\label{eq':LD2S_LP_zx} & \max\{z^{u}_{v}, z^{v}_u\} \leq x_{\aset{u,v}} & \forall \{u,v\} \in E\\
\label{eq':LD2S_LP_degree}& \sum_{v \in \Gamma(u)} x_{\aset{u,v}} \leq \lambda & \forall u \in V \\
\label{eq':LD2S_LP_cover} & x_{\aset{u,v}} + \sum_{w \in \Gamma(u) \cap \Gamma(v)} z^{w}_{\{u,v\}} = 1& \forall \{u,v\} \in \hat E
\end{align}

Constraint~\eqref{eq':LD2S_LP_SmES} requires that for each neighborhood graph $G_u$ there is an associated fractional \SmES solution.%
\footnote{\label{foot:guess}%
Our actual relaxation (Figure~\ref{fig:ld2s-lp}) has a collection of \SmES instances for each neighborhood graph based on the possible degrees in a bipartite decomposition of an optimal solution, and we allow the LP to fractionally ``guess'' which of these instances to use.}  Constraint~\eqref{eq':LD2S_LP_zx} simply requires that for each edge, if either of the \SmES instances at its endpoints include it in their solution then we include it in the overall solution.  Constraint~\eqref{eq':LD2S_LP_degree} gives the degree bound, and~\eqref{eq':LD2S_LP_cover} is the main covering constraint, requiring that every demand is either included or is spanned by a $2$-path.  It is easy to see that this is a valid relaxation for \LD2S: if we are given a $2$-spanner $H$ of $G$ with maximum degree at most $\lambda$, for every edge $\{u,v\} \in E(H)$ we set $x_{\{u,v\}} = 1$ and $z^v_u = 1$ and $z^u_v =1$.  For every edge $\{u,v\} \in E \setminus E(H)$ we arbitrarily choose some $w \in V$ so that $\{u,w\} \in E(H)$ and $\{w,v\} \in E(H)$ (some such $w$ must exist since $H$ is a $2$-spanner) and set $z^w_{\{u,v\}} = 1$.  All other variables are $0$.  %

We now show that it is sufficient to design a rounding scheme for \SmES that is faithful according to the following definition.  Given a graph $G$, let $\mathcal L(G)$ be an LP that has a variable $\zeta_u$ for every $u \in V(G)$ and a variable $\zeta_e$ for every $e \in E(G)$ (we will later instantiate $\mathcal L(G)$ as various LP relaxations of \SmES).

\begin{definition} \label{def:faithful}
A randomized rounding \algA is a \emph{factor $f$ faithful rounding} for $\mathcal L(G)$ if, when given a valid solution $((\zeta_u)_{u \in V}, (\zeta_e)_{e \in E})$ to $\mathcal L(G)$, it produces a randomized (not necessarily induced) subgraph $H^*=(V^*,E^*)$ such that
\begin{enumerate}
\compactify
\item \label{def:faithful_prob_vertex} $\Pr[v \in V^*] \leq f \cdot \zeta_v$ for all $v \in V(G)$,
\item \label{def:faithful_prob_edge} $\Pr[\{u,v\}\in E^*] \leq \zeta_{\{u,v\}}$ for all $\{u,v\} \in E(G)$,
\item \label{def:absoluteVE-bounds}
$|V^*| \leq f\cdot \sum_{v\in V(G)} \zeta_{v}$ \, (with probability 1), and
\item \label{def:expectedE-LB}
$\E[|E^*|] \geq \tOmega(\sum_{\{u,v\} \in E(G)} \zeta_{\{u,v\}})$.
\end{enumerate}
\end{definition}

Observe that if \algA is a factor $f$ faithful rounding for a relaxation of \SmES
then it is also an $f$-approximation in the usual sense,
simply by conditions~\ref{def:absoluteVE-bounds} and~\ref{def:expectedE-LB} (up to a polylogarithmic loss in the amount of edges covered).
The converse, however is not true: many rounding algorithms that give an $f$-approximation are not faithful, including~\cite{BCCFV10}.

We now show that if we are given an \algA that is a factor $f(n)$ faithful rounding for \SmES (where $n$ is the number of vertices in the \SmES instance), there is an $\tO(f(\Delta))$-approximation algorithm for \LD2S that uses \algA as a black box. The reduction is given as Algorithm~\ref{alg':LD2S}. It begins with all edges as the demand set $\hat E$, and first solves the LP relaxation for \LD2S with demand set $\hat E$.  It adds every edge that has $x$ value at least $1/4$, and then uses \algA to round each of the $|V|$ \SmES instances in the relaxation.   At the end of the loop it updates the demands $\hat E$ by removing edges that were successfully covered by this process, and repeats.  Note that the \emph{edges} covered by the \SmES roundings are used only in the analysis; in the algorithm we take the \emph{vertices} output by each \SmES solution and include the appropriate edges in our spanner.

\begin{algorithm}[]
\caption{Approximation algorithm for \LD2S}
\label{alg':LD2S}
\SetKwInOut{Input}{Input}\SetKwInOut{Output}{Output}
\Input{Graph $G = (V, E)$, degree bound $\lambda$, factor $f(n)$ faithful rounding \algA for \SmES}
\Output{$2$-spanner $H = (V, E_H)$ of $G$}
$\hat E \leftarrow E$, $E_H\leftarrow \emptyset$
\\
\While{$\hat E \neq \emptyset$}{
	Compute a valid solution $\tuple{\vec x,\vec z}$ for
        LP~\eqref{eq':LD2S_LP_SmES}-\eqref{eq':LD2S_LP_cover} on graph $G$ with demands $\hat E$
	\\
	$E_x \leftarrow \{e \in E : x_e \geq 1/4\}$
	\\
	\ForEach{$u \in V$}{
		$H^*_{u} \leftarrow \mathcal A(G_u,\vec z^u)$ \label{alg':applyA}
			\tcp*{output of \SmES rounding $\mathcal A$}	
					
			$E_{u} \leftarrow \{\{u,v\} \in E : v \in V(H^*_{u})\}$
	}
	\tcp{Add all edges found in above rounding}
	$E_H \leftarrow E_H \cup E_x \cup\left(\bigcup_{u \in V} E_{u} \right)$
	\\
	\tcp{Remove satisfied demands}
        $\hat E \leftarrow \hat E \setminus\left(E_H \cup 
        \big\{ \aset{u,v} : \exists w\in V \textrm{ s.t. both } \{u,w\},\{w,v\} \in E_H 
        \big\} \right)$
}
\end{algorithm}

\begin{theorem} \label{thm':LD2StoSmES}
Let \algA be a factor $f(n)$ faithful rounding for \SmES (where $n$ is the number of vertices in the \SmES instance).
Then there is a (randomized) $\tO(f(\Delta))$-approximation for \LD2S.
\end{theorem}
\begin{proof}
We provide only a sketch of the proof; details can be found in Section~\ref{sec:reduction}.  We may assume that our \LD2S algorithm guesses some $\lambda\in[\OPT,2\cdot\OPT]$ simply by trying the $O(\log \Delta)$ relevant values and reporting the best solution.  In this case, LP \eqref{eq':LD2S_LP_SmES}-\eqref{eq':LD2S_LP_cover} is guaranteed to have a feasible solution.  We now use Algorithm~\ref{alg':LD2S} with this value of $\lambda$.  It is easy to see that each iteration of the loop only increases the maximum degree by $\tO(f(\Delta)) \cdot \OPT$: adding edges in $E_x$ only costs a constant factor more than the fractional solution, Definition~\ref{def:faithful}(\ref{def:absoluteVE-bounds}) implies that rounding the \SmES solution at $u$ only increases the degree of $u$ by $f(|V_u|) \cdot \OPT\leq f(\Delta)\cdot \OPT$, and Definition~\ref{def:faithful}(\ref{def:faithful_prob_vertex}) implies that rounding the \SmES solution at neighbors of $u$ only increases the degree of $u$ by $\tO(f(|V_u|)) \cdot \OPT$ (with high probability).

So we just need to show that the number of iterations is (with high probability) at most $\tO(1)$.  To do this we prove that in every iteration the expected number of satisfied demands is at least $\tOmega(|\hat E|)$.  This is clearly true if $|E_x|$ is large.  If $|E_x|$ is small, then summing \eqref{eq':LD2S_LP_cover} over all $\{u,v\} \in \hat E$ implies that the total amount of demand covered by \SmES instances (i.e.~the $z$ variables) is large, say $\Omega(|\hat E|)$.  Now Definition~\ref{def:faithful}(\ref{def:expectedE-LB}) guarantees that when we round an instance we still cover almost as much demand as the LP (up to a polylogarithmic factor), and Definition~\ref{def:faithful}(\ref{def:faithful_prob_edge}) implies that this coverage is spread out among the demands in a way that corresponds to the LP.  This, together with~\eqref{eq':LD2S_LP_cover}, implies that the total coverage is large but no edge is ``overcovered''.  Thus while we cannot guarantee that \emph{every} demand is covered with \emph{high} probability in each iteration, an averaging argument guarantees that \emph{many} demands are covered with \emph{reasonable} probability.  This is enough to give the expected amount of coverage that we need.
\end{proof}

\subsection{Overview of our faithful rounding algorithm for \SmES}
In this section we describe our faithful factor $n^{3-2\sqrt{2}+\eps}$ rounding algorithm for \SmES. While the description of the full algorithm is rather lengthy, and therefore deferred to Sections~\ref{sec:SmES-algorithm},~\ref{sec:SmES-details}, and~\ref{sec:SmES-analysis}, we give a high-level overview (with some technical details) and concentrate on a special case which illustrates the main ideas in the algorithm and its analysis.

\subsubsection{{\sc Densest $k$-Subgraph} and the log-density framework}\label{sec:DkS-overview}
We follow the framework introduced in~\cite{BCCFV10}.  They begin by defining the notion of log-density of a graph as $\log_n(D_{\mathrm{avg}})$, where $D_{\mathrm{avg}}$ is the average degree and $n$ is the number of nodes.  They then asked the following question: how hard is it to distinguish between 1) a random graph, and 2) a graph containing a subgraph with roughly the same log-density as the first graph?

More formally, they pose the following {\sc Dense versus Random} promise problem, parameterized by $k$ and constants $0<\alpha,\beta<1$: given a graph $G$, distinguish between the following two cases:
\begin{enumerate}
\item $G=G(n,p)$ where $p=n^{\alpha-1}$ (this graph has log-density concentrated around $\alpha$). %
\item $G$ is adversarially chosen so that the densest $k$-subgraph has log-density $\beta$ (where $k^{1+\beta}\gg pk$).
\end{enumerate}

For certain ranges of parameters, it seems quite challenging to efficiently distinguish when $\beta<\alpha$. In fact, the following hypothesis is %
 consistent with the current state of our knowledge:
\begin{hypothesis}\label{hyp:dense-v-random} For all $0<\alpha<1$, for all sufficiently small $\eps>0$, and for all $k\leq\sqrt{n}$, we cannot solve {\sc Dense versus Random} in polynomial time (w.h.p.) when $\beta\leq\alpha-\eps$.
\end{hypothesis}

The above hypothesis (if true) has immediate implications for the hardness of approximation of both \DkS and \SmES. Concretely, for \SmES, let $m=k^{1+\beta}$ be the number of edges in $k$-subgraph in the second case. We know that in the first case w.h.p.\ the smallest $m$-edge subgraph has size at least $\widetilde{\Omega}(\min\{m,\sqrt{mn^{1-\alpha}}\})$. Thus, if we could achieve approximation ratio $\ll k/\min\{m,\sqrt{mn^{1-\alpha}}\}$, this would refute Hypothesis~\ref{hyp:dense-v-random} for the corresponding parameters. %
For %
$k=n^{\sqrt{2}-1}$ and $\alpha=\sqrt{2}-1$, the hypothesis implies that there exists no $n^{3-2\sqrt{2}-\eps}$-approximation for \SmES.

While~\cite{BCCFV10} matches the gap predicted by the log-density model for \DkS with an $n^{1/4+\eps}$ approximation for \DkS (even for general graphs), we also match the predicted gap for \SmES with an $n^{3-2\sqrt2+\eps}$-approximation for \SmES.

\subsubsection{Parametrization and simplifications}
In order to achieve a faithful rounding, we will make certain assumptions (which we later justify) about the structure of the intended solution to the LP relaxation. In particular, we will assume that the subgraph represented by the solution is regular, and that we are allowed to ``guess'' the size of the subgraph, $k$, and the degree in the subgraph, $d$, thus $m=\Theta(kd)$ (see Section~\ref{sec:SmES_LP}).%

We also make the following simplifying assumptions. Let $f=f(n,k,d)$ be the intended approximation factor, which will be determined shortly. 
We may assume that $f\leq d$, since it is easy to achieve a faithful $O(d)$-approximation (see Appendix~\ref{sec:small-degrees}). %
We also assume that the maximum degree in the input graph is at most $D=nd/(kf^2)$ (see Appendix~\ref{sec:max-deg-bound}).%

Finally write $\alpha=\log_n(D)$, and define our intended approximation $f$ implicitly as the value which satisfies $f=n^{\alpha(1-\alpha)/(1+\alpha)}$ (together with the definition of $D$ we can derive an explicit expression for $\alpha$ and $f$). Note that maximizing this expression over $\alpha\in[0,1]$ shows that
$f\leq n^{3-2\sqrt2}$.

\subsubsection{LP relaxation and faithful rounding for \SmES}
With the previous assumptions in mind, we have the following feasibility-LP relaxation (simplified for this overview) which is implied by $q$ rounds of Sherali-Adams, with variables $\{z_T\mid T\subset V\cup E\text{, }|T|\leq q\}$ (in the intended 0-1 solution, $z_T=1$ if and only if all vertices and edges in $T$ are in the subgraph):

\begin{align}
\label{eq:simple-k-bound}
&\sum_{v \in V} z_{T \cup \{v\}} = k z_T &\forall T\subset V\cup E\text{, }|T|\leq q-1 \\
\label{eq:simple-d-bound} & \sum_{u \in \Gamma(v)} z_{T \cup \{u\}} = d z_T &  \forall T\subset V\cup E\text{, }|T|\leq q-1, \forall v \in T \cap V \\
\label{eq:simple-smes-consistency} & z_T = z_{T\cup\{u\}}=z_{T\cup\{v\}}=z_{T\cup\{u,v\}} & \forall T\subset V\cup E\text{, }|T|\leq q-2, \forall \{u,v\} \in T\cap E\\
\label{eq:simple-smes-monotonicity}
& 0 \leq z_T \leq z_{T'}\leq z_\emptyset=1 & \forall T' \subseteq T
\end{align}

The algorithm in its full generality is based on the caterpillar structures introduced in~\cite{BCCFV10} (where the caterpillar structure depends on $\alpha$). Let us concentrate here on the case where $\alpha=1/s$ for some (fixed) integer $s>0$, in which case the caterpillar is simply a path of length $s$. At its core, the algorithm (for this value of $\alpha$) relies on an LP-analogue of the following combinatorial argument. Fix a vertex $v_0$ in the optimum subgraph. For all $t=1,\ldots, s$, let $P^{v_0}_t$ be the union of all (possibly self-intersecting) paths of length $t$ in the subgraph starting at $v_0$, and let $V^{v_0}_t$ be final endpoints of those paths. Note that $|V^{v_0}_1|=d$ and that $|V^{v_0}_s|\leq k=\frac{dn}{f^2D}=\frac{d}{f^2}\cdot n^{1-\alpha}=\frac{d}{f^2}f^{(1+\alpha)/\alpha}=df^{s-1}$. Therefore, there must be some $t\in\{1,\ldots,s-1\}$ for which $|V^{v_0}_{t+1}|/|V^{v_0}_t|\leq f$. Now consider the subgraph at this step $H^{v_0}_t=(V^{v_0}_t,V^{v_0}_{t+1},\{\{v_t,v_{t+1}\}\mid \exists v_0-\ldots-v_t-v_{t+1}\in P^{v_0}_{t+1}\})$.  Since the vertices in $V^{v_0}_t$ all have degree $d$, the average degree of vertices in $V^{v_0}_{t+1}$ is at least $d/f$. It turns out that even without access to the optimum subgraph we can isolate a subgraph with average degree at least $d/f$ and at most $kf$ vertices (this is essentially because by the degree bound, the number of vertices at any intermediate stage is at most $D^{s-1}=n^{1-\alpha}=\frac{n}{D}=k\cdot\frac{f^2}{d}\leq kf$). This essentially gives an $f$-approximation for \SmES (since we can repeat until accumulating $m$ edges).

Here we come to the fundamental difficulty in adapting such an approach to achieve a faithful rounding. The combinatorial algorithm depends on choosing an initial vertex $v_0$ which is actually in the optimum subgraph. 
The analogous LP-rounding algorithm uses the LP values ``conditioned on choosing $v_0$'', that is, values of the form $z_{S\cup\{v_0\}}/z_{v_0}$ instead of the original $z_S$ variables 
(where $S$ corresponds to one or more vertices/edges along the path). 
However, it is the $z_S$ variables (in particular for singleton sets $S$ representing one vertex or one edge) which we want to be faithful to in our rounding.\footnote{This problem is only exacerbated in the general case, when the caterpillar has additional leaves to condition on.} Unfortunately, these two LP solutions might be almost completely unrelated.

To overcome this difficulty, we use a somewhat elaborate bucketing scheme, to ensure that all the relevant LP values are reasonably uniform, as follows. 
Denote by $\mathcal{P}^{v}_t$ the set of all length $t$ paths in the graph starting at vertex $v$, and by $z_p$ the variable for a path $p$ 
(i.e., $z_T$ where $T$ is the set of edges and vertices in $p$, 
or by Constraint~\eqref{eq:simple-smes-monotonicity}, 
$T$ could equivalently be just the edges in~$p$).
The core of the analysis of the LP-analogue relies on the equality
$$\sum_{v}\sum_{p\in\mathcal{P}^v_t}z_p=\sum_{v}d^tz_{\{v\}}=kd^t,$$
obtained by Constraint~\eqref{eq:simple-k-bound} and repeated applications of~\eqref{eq:simple-d-bound},
but in fact it can use any set of length-$s$ paths $\mathcal{P}$ for which $\sum_{p\in\mathcal{P}}z_p=\widetilde{\Omega}(kd^s)$. %
Thus by partitioning the set of paths $\bigcup_vP^v_s$ into buckets 
and choosing a bucket $\mathcal P$ with the largest LP value, we can ensure that in every path $p=u_0-u_1-\ldots-u_s$ in the bucket $\mathcal P$ certain LP values (like the ones corresponding to entire paths, $z_p$, or the ones corresponding to path prefixes, $z_{\{\{u_{i-1},u_{i}\}\mid i\in[t]\}}$ for some $t\in[s-1]$, or to vertices in certain positions, $z_{\{u_t\}}$, or to ``conditioned'' values, $z_{\{u_0,u_t\}}/z_{\{u_0\}}$) are all independent of the choice of path (up to a constant factor). In other words, within the bucket $\mathcal P$ (say, vertices $u_t$ for a fixed $t\in\{0,\ldots,s\}$), the corresponding LP values will be essentially uniform over the choice of starting vertex $u_0$ and path~$p$.

Using the uniformity obtained via the above bucketing scheme, 
we can relate the algorithm (which is based on the conditioned LP values) 
to the original LP values.
After some additional combinatorial bucketing, we can run the following algorithm: let $\mathcal{V}_0$ be the set of starting vertices $u_0$ (i.e.~paths of length $0$) that survive the bucketing, pick a starting vertex $u_0\in\mathcal{V}_0$ uniformly at random, and for whichever level $t\in[s-1]$ that gives the approximation guarantee (it can be shown that such a $t$ exists), output the level $t$ subgraph $H^{u_0}_t=\{\{u_t,u_{t+1}\}\mid \exists p=u_0-u_1-\ldots-u_s\in\mathcal{P}\}$. Since LP values are uniform, 
the question essentially becomes, how do we guarantee that no bucketed vertex (or edge) is chosen with much higher probability than the rest (or the average)? This is where we crucially use the regularity Constraint~\eqref{eq:simple-d-bound} (as opposed to, say, a minimum degree constraint, as in~\cite{BCCFV10}). Roughly speaking, individual vertices and edges cannot be reached by a disproportionately large fraction of vertices $u_0\in\mathcal{V}_0$, because then the relative total LP weight of the corresponding paths (to such a vertex or edge) would exceed $d^t$.

For the sake of concreteness, let us consider one specific aspect of faithful rounding: the probability with which the level $t$ vertices $u_t$ are chosen. %
Let $\mathcal{P}_t$ be the set of length $t$ prefixes of paths in $\mathcal{P}$,
let $\mathcal{P}^{u_0}_t$ be the set of paths in $\mathcal{P}_t$ that start with 
the vertex $u_0\in\mathcal{V}_0$, 
and let $\mathcal{V}_t$  (resp.\ $\mathcal{V}^{u_0}_t$) be the set of level $t$ endpoints of paths in $\mathcal{P}_t$ (resp.\ in $\mathcal{P}^{u_0}_t$).
 By the approximation guarantee (via an LP analogue of the above combinatorial argument), we have
\begin{equation}\label{eq:smes-overview-comb-bound} |\mathcal{V}^{u_0}_t|\leq fk.
\end{equation}
Suppose the bucketing also ensures that every $u_0\in\mathcal{V}_0$ and $u_t\in\mathcal{V}^{u_0}_t$ are connected by roughly the same number of $\mathcal{P}_t$ paths (up to a constant factor), which we denote by $h$. Also, suppose the cardinalities $|\mathcal{V}^{u_0}_t|$ are roughly uniform for different choices of $u_0$.
Then, abusing notation, we can write the number of paths as $|\mathcal{P}_t|\approx |\mathcal{V}_0|\cdot|\mathcal{V}^{u_0}_t|\cdot h$, and in particular, the total weight of paths $p\in\mathcal{P}_t$ is $z_p|\mathcal{V}_0|\cdot|\mathcal{V}^{u_0}_t|\cdot h\approx kd^t$. %
Now, by repeated applications of Constraint~\eqref{eq:simple-d-bound}, we have that the total weight of paths leading to a specific vertex $u_t\in\mathcal{V}_t$ is $z_p|\{u_0\mid u_t\in\mathcal{V}^{u_0}_t\}|h\leq z_{u_t} d^t$ (note that this argument reverses the direction of paths in the algorithm and so crucially depends on the existence of consistent high-moment Sherali-Adams variables, which are not present in the \Lovasz-Schrijver hierarchy used in \cite{BCCFV10}). Combining this with the (approximate) equality above, we can bound the probability that a vertex $u_t$ is included in the output (the level $t$ subgraph) as
\begin{equation*}
\frac{|\{u_0\mid u_t\in\mathcal{V}^{u_0}_t\}|}{|\mathcal{V}_0|} \leq \frac{d^tz_{u_t}}{z_ph|\mathcal{V}_0|} =\frac{|\mathcal{V}^{u_0}_t|z_{u_t}}{k} \leq fz_{u_t},\text{ by~\eqref{eq:smes-overview-comb-bound}.}
\end{equation*}

\section{LP relaxations} \label{sec:LPs}

In this section we develop the basic LP relaxations that we will use, both for \SmES and for \LD2S.  We begin with \SmES, since we will need the LP we develop in order to define the LP for \LD2S.

\subsection{LP relaxation for \SmES} \label{sec:SmES_LP}

It turns out that it is easier to develop faithful rounding algorithms for \SmES if we make certain simplifying assumptions. Namely, we would like to assume that the input graph is bipartite, and that the optimal solution is nearly-regular
(vertices on the same side of the bipartition have degree within an $O(\log n)$ factor of each other).  These assumptions will affect our relaxation, so we discuss them here.  Since these assumptions involve manipulations of the optimal (or at least an unknown) subgraph, one should view these as a thought experiment which justifies the correctness (i.e.\ feasibility) of our relaxations.

Formally, we define nearly-regular as follows:
\begin{definition}
A bipartite graph $G = (U_0, U_1, E)$ is called \emph{$(k_0, k_1, d_0, d_1)$-nearly regular} if for every $b\in[2]$ we have $|U_b| = k_b$ and the following condition on the degrees holds:
$$
d_b \geq \max_{u_b\in U_b}\deg(u_b) \ge \min_{u_b\in U_b}\deg(u_b)
  \ge\Omega(d_b/\log n)
$$
\end{definition}
By convention, we will assume that $k_0\geq k_1$, which implies that $d_1\geq \Omega(d_0 / \log n)$. %

The next lemma shows that any graph $H$ can be changed into a nearly-regular graph with almost the same number of edges.  In particular, any dense subgraph $H$ can be made nearly regular without losing too much in the density.

\begin{lemma} \label{lem:regularity_reduction}
Given an arbitrary graph $H = (V,E)$ on $n$ vertices, there exist values
$k_0, k_1, d_0, d_1$ and disjoint vertex sets $U_0,U_1 \subseteq V$ with the properties that the induced bipartite subgraph $H'$ of $H$ on $(U_0, U_1)$ is $(k_0, k_1, d_0, d_1)$-nearly regular and also has $|E(U_0, U_1)| \geq \Omega(1/\log^2 n) |E|$.
\end{lemma}
\begin{proof}
First, note that a random cut in the graph gives a bipartition which preserves at least half the edges (in expectation). Thus there exists a bipartition $(V_0,V_1)$ with $|E(V_0, V_1)| \geq |E|/2$.  For the rest of the proof we will only use these edges between $V_0$ and $V_1$.  Now,
partition the vertices in $V_0$ into $\log n$ buckets by degree (into $V_1$), where in each bucket, degrees vary up to a factor of at most $2$. Let $\tilde{U}_0$ be the bucket which sees the most edges (at least $|E|/2\log n$). Let $d_0$ be the maximum degree in $\tilde{U}_0$. Note that the average degree in $\tilde{U}_0$ is at least $d_0/2$.

Now partition the vertices of $V_1$ into buckets by their degree into $\tilde{U}_0$. Let $\tilde{U}_1$ be the bucket with the largest number of edges into $\tilde{U}_0$ (at least $|E|/2\log^2 n$). Let $d_1$ be the maximum degree (into $\tilde{U}_0$) in $\tilde{U}_1$, and note that the average degree in $\tilde{U}_1$ is at least $d_1/2$. Since the number of edges in this step went down by a most a $\log n$ factor, the average degree in $\tilde{U}_0$ (into $\tilde{U}_1$) is at least $d_0/(2\log n)$. If we now iteratively remove every vertex in $\tilde{U}_0$ with degree (into $\tilde{U}_0$) at most $d_0/(6\log n)$, and every vertex in $\tilde{U}_1$ with degree (into $\tilde{U}_0$) at most $d_1/6$, then it is easy to see that at least a $1/3$-fraction of the original edges in $E(\tilde{U}_0,\tilde{U}_1)$ remains, and this final step guarantees the regularity conditions on both sides.
\end{proof}

Before defining our relaxation for general graphs, let us first define a feasibility LP relaxation for the decision problem of whether, given a %
bipartite graph $G'=(V_0,V_1,E)$, %
the graph contains a $(k_0, k_1, d_0, d_1)$-nearly regular subgraph, with $k_0$ vertices on the $V_0$ side, and $k_1$ vertices on the $V_1$ side. For any integer $q$, let $\mathcal T_q =\mathcal T_q(G') = \{T \subseteq V_0 \cup V_1 \cup E':\, |T| \leq q\}$.

We denote by ${\textbf{Bipartite-SmES-LP}_q(G', k_0, k_1, d_0, d_1)}$
the set of solutions to a feasibility LP
given by Constraints~\eqref{eq:smes1}-\eqref{eq:smes6},
over the variables $(y_T)_{T \in \mathcal T_q}$,
as depicted in Figure~\ref{fig:bipartite-smes-lp}. %
These constraints are actually implied by $q$ rounds of Sherali-Adams applied
to a basic \SmES LP that has variables for all vertices \emph{and} edges.

\begin{figure}[htb]
\Hrzline%
\caption{Relaxation $\textbf{Bipartite-SmES-LP}_q(G', k_0, k_1, d_0, d_1)$
on variables $(y_T)_{T \in \mathcal T_q}$}
\label{fig:bipartite-smes-lp}
\hrulefill
\begin{align}
\hrulefill
\label{eq:smes1}
&\sum_{v \in V_b} y_{T \cup \{v\}} = k_b y_T &\forall b \in \{0,1\}, \forall T \in \mathcal T_{q-1} \\
\label{LP:d-bound} & \Omega(\tfrac{1}{\log n})\, d_b y_T \leq \sum_{u \in \Gamma(v)} y_{T \cup \{\{u,v\}\}} \leq d_b y_T & \forall b \in \{0,1\}, \forall T \in \mathcal T_{q-1}, \forall v \in T \cap V_b \\
\label{eq:smes3} & y_T = y_{T\cup\{u\}}=y_{T\cup\{v\}}=y_{T\cup\{u,v\}} & \forall T \in \mathcal T_{q-2}, \forall \{u,v\} \in T\\
\label{eq:smes6} & 0 \leq y_T \leq y_{T'}\leq 1 & \forall T' \subseteq T \in \mathcal T_q%
\end{align}
\Hrzline%
\end{figure}

\paragraph{Remark.}
Normally, such an LP also includes the normalization $y_\emptyset = 1$.
However, for our intended usage of this LP, namely for \LD2S,
it will be important to leave this variable unconstrained,
except for the upper bound given by~\eqref{eq:smes6}.  But when we round this LP we will be able to assume that $y_{\emptyset} = 1$; whenever we choose to round it we will also scale it by $1/y_{\emptyset}$.  Thus when we consider faithful rounding algorithms for this LP we will always be assuming that $y_{\emptyset} = 1$.

\paragraph{Remark.}
This LP has size (variables and constraints) at most $n^{O(q)}$,
because the number of variables is dominated by $O(|\mathcal T_{q}|)$,
and the number of constraints is dominated by $O(|\mathcal T_{q}|^2)$.

Notice that according to Constraint~\eqref{eq:smes3}, an edge is covered
if and only if both endpoints are chosen \emph{and} we decide to take the edge.
For example, it is feasible to have $y_{\aset{u}}=y_{\aset{v}}=y_{\{u,v\}} = 1$ (i.e.~both $u$ and $v$ are chosen) but $y_{\{\{u,v\}\}} = y_{\{u,v,\{u,v\}\}} = 0$ (i.e.~for some reason the LP does not count this edge as being covered).  Normally for covering problems such as \SmES there is no reason not to include an edge if both vertices are included, but because of how we use \SmES in our algorithm for \LD2S it will be important for us to be able to include two vertices without necessarily covering the edge between them.

Observe that if $G'$ %
contains a $(k_0, k_1, d_0, d_1)$-nearly regular subgraph $H$,
then the above LP %
 has a feasible solution
that corresponds to $H$ in the sense that
$y_T=1$ if and only if every vertex and edge in $T$ is present in $H$ (and $y_\emptyset=1$).

Now, consider a (not necessarily bipartite) graph $G$.  Lemma~\ref{lem:regularity_reduction} implies that for any subgraph $H$ of $G$ there is a nearly-regular bipartite subgraph of $H$ with almost as many edges.  However, since we do not know a priori the bipartition of the subgraph (or any bipartition of $G$ which is consistent with it), we write a ``container'' LP whose purpose is to essentially assign sides, and to interface cleanly between the bipartite \SmES relaxation, and the relaxation for \LD2S. In brief, with every vertex $v\in V$, we associate at most one of two assignments, which we label $(v,1)$ and $(v,2)$. Thus, every edge $\{u,v\}$ can participate in the subgraph as (at most) one of two possible ``assigned'' edges $\{(u,1),(v,2)\}$ and $\{(u,2),(v,1)\}$.
We can think of these as forming a bipartite graph $B(G)$ with vertices $V \times [2]$ and edges $\{(u,1),(v,2)\}$ and $\{(u,2),(v,1)\}$ for every $\{u,v\} \in E$.  The induced nearly-regular bipartite subgraph $H'$ of $G$ that is guaranteed to exist by Lemma~\ref{lem:regularity_reduction} corresponds to an induced subgraph of $B(G)$ in which for every vertex $v$ in $H'$ either $(v,1)$ or $(v,2)$ is included (depending on the side of $v$ in $H'$).  Note that for every vertex (or edge) in $H'$, there is exactly one vertex (or edge) in the corresponding subgraph of $B(G)$.

It will be easier to describe our algorithm and analysis for \SmES using the LP for bipartite graphs, but in order to interface with the \LD2S LP we will use the above discussion of $B(G)$ to write a container or wrapper LP.
Denote by ${\textbf{SmES-LP}_q(G, k_0, k_1, d_0, d_1)}$
the set of solutions to the feasibility LP depicted in Figure~\ref{fig:smes-lp}
on the variables $(z_a)_{a \in V \cup E}$ and $z_{\emptyset}$ (in other words there is a $z$ variable for every vertex and edge in the original graph, as well as one extra $z$ variable for the empty set).
Notice that this LP has auxiliary variables of the form $y_T$.

\begin{figure}[htb]
\Hrzline%
\caption{Relaxation $\textbf{SmES-LP}_q(G, k_0, k_1, d_0, d_1)$
on variables $(z_a)_{a \in V \cup E\cup\aset{\emptyset}}$}
\label{fig:smes-lp}
\hrulefill
\begin{align}
  \exists (y_T&)_{T\in\mathcal T_q(B(G))}\in \, {\textbf{Bipartite-SmES-LP}_q(B(G), k_0, k_1, d_0, d_1)} \text{ s.t.} \nonumber\\
\label{eq:smes_z_empty} & z_{\emptyset} = y_{\emptyset} \\
\label{eq:smes_zu} & y_{\{(u,1)\}} + y_{\{(u,2)\}} = z_u\leq z_{\emptyset}& \forall u \in V\\
\label{eq: smes_zuv} & y_{\{\{(u,1), (v,2)\}\}}+ y_{\{\{(u,2), (v,1)\}\}} = z_{\{u, v\}}\leq z_u,z_v & \forall \{u,v\} \in E
\end{align}
\Hrzline%
\end{figure}

It is easy to see that for any $(k_0, k_1, d_0, d_1)$-nearly regular bipartite subgraph $H$ of $G$, if we set $z_u = 1$ for $u \in V(H)$ and $z_e = 1$ for $e \in E(H)$ then we can set the $y$ variables in a way corresponding to the associated subgraph of $B(G)$, giving a valid solution.  For any tuple of parameters $\tau = \langle k_0, k_1, d_0, d_1 \rangle$ let $\textbf{SmES-LP}^{\tau}_q(G) = \textbf{SmES-LP}_q(G, k_0, k_1, d_0, d_1)$ and let $\textbf{Bipartite-SmES-LP}^{\tau}_q(G') = \textbf{Bipartite-SmES-LP}_q(G', k_0, k_1, d_0, d_1)$.  By construction, a faithful rounding for $\textbf{Bipartite-SmES-LP}^{\tau}_q(B(G))$ (according to Definition~\ref{def:faithful}) implies a faithful algorithm for $\textbf{SmES-LP}^{\tau}_q(G)$; we now prove this.

\begin{lemma} \label{lem:faithful_from_bipartite}
For any graph $G$ and parameters $\tau$, if we have a factor $f$ faithful rounding algorithm for $\textbf{Bipartite-SmES-LP}^{\tau}_q(B(G))$ then we have a factor $f$ faithful rounding algorithm for $\textbf{SmES-LP}^{\tau}_q(G)$.
\end{lemma}
\begin{proof}
Let \algA be a factor $f$ faithful rounding for ${\textbf{Bipartite-SmES-LP}^{\tau}_q(B(G))}$.  Then our algorithm for rounding $\textbf{SmES-LP}^{\tau}_q(G)$ is simple: we first find the associated $y$ variables (if they are not already given to us), and then run \algA on the associated $y$ variables to get $(\tilde V^*, \tilde E^*)$.  We then include a vertex $v$ in $V^*$ if $\tilde V^*$ includes either $(v,1)$ or $(v,2)$, and include an edge $\{u,v\}$ in $E^*$ if $\tilde E^*$ includes either $\{(u,1) (v,2)\}$ or $\{(u,2), (v,1)\}$.

Then $\Pr[v \in V^*] \leq f \cdot y_{\{(v,1)\}} + f \cdot y_{\{(v,2)\}} = f \cdot z_v$, satisfying the first part of the definition.  And $\Pr[\{u,v\} \in E^*] \leq y_{\{\{(u,1), (v,2)\}\}} + y_{\{\{(u,2), (v,1)\}\}} = z_{\{u,v\}}$, satisfying the second part.  For the third part, we know that $|\tilde V^*|$ is (with probability $1$) at most $f \cdot \sum_{v \in V(G)} (y_{\{(v,1)\}} + y_{\{(v,2)\}}) = f \cdot \sum_{v \in V(G)} z_v$ vertices, and clearly $|V^*| \leq |\tilde V^*|$.  Finally, we know that $\E[|E^*|] \geq \E[|\tilde E^*|] / 2 \geq \tOmega(\sum_{\{u,v\} \in E(G)} (y_{\{\{(u,1), (v,2)\}\}} + y_{\{\{(u,2), (v,1)\}\}})) = \tOmega(\sum_{\{u,v\} \in E(G)} z_{\{u,v\}})$
\end{proof}

\subsection{LP relaxation for \LD2S} \label{sec:LD2S_LP}

Now we show how to use this \SmES LP to give an LP relaxation for \LD2S.  We will actually give a relaxation for a slightly more general version of \LD2S in which instead of spanning \emph{all} edges we are given a subset $\hat E \subseteq E$ and are only required to span edges in $\hat E$.  Note that the optimal solution for demands $\hat E \subseteq E$ has maximum degree that is at most the maximum degree of the optimal solution to the original \LD2S problem (where all edges are demands).  This will allow us to cover some demands, re-solve the LP with only the remaining demands, and repeat.

Since we construct a feasibility LP we will guess the optimal degree bound $\lambda$, so we will be able to treat it like a constant (rather than a variable).
For each $u \in V$, let $G_u = (V_u, E_u)$ be the induced subgraph of $(V, \hat E)$ on $\Gamma_G(u)$, i.e.~$G_u$ has vertex set $\Gamma_G(u)$ (the neighbors of $u$ in the overall graph) and edge set $\hat E$ restricted to $\Gamma_G(u)$.
The core of our relaxation is a decomposition of one \SmES instance into many \SmES instances.  In particular, the following lemma will be useful:

\begin{lemma} \label{lem:SmES_decomposition}
Given a graph $G = (V,E)$ on $n$ vertices and a value $\lambda$, there is a multiset $L^{\lambda}$ of size at most $O(n^4 \log^3 n)$ consisting of tuples $\tuple{k_0, k_1, d_0, d_1}$ so that every subgraph of $G$ with at most $\lambda$ vertices can be decomposed into $O(\log^3 n)$ nearly-regular subgraphs with parameters from $L^{\lambda}$.
\end{lemma}
\begin{proof}
Let $H$ be a subgraph of $G$ with at most $\lambda$ vertices.  We know from Lemma~\ref{lem:regularity_reduction} (applied to $H$) that there exists some $\tau = \tuple{k_0, k_1, d_0, d_1}$ and subgraph $H'$ of $H$ so that $H'$ is nearly-regular with parameters $\tau$ and the number of edges in $H'$ is at least an $\Omega(1/\log^2 n)$-fraction of the number of edges in $H$.  We can now remove the edges in $H'$ from $H$ and repeat this process.  Since there are at most $n^2$ edges initially, we can repeat this step only $O(\log^3 n)$ times.  Note that this bound is independent of the graph $H$ that we are decomposing.  $H$ simply affects \emph{which} tuples are produced by the decomposition.  But no matter what $H$ is, obviously $k_0, k_1, d_0, d_1$ are all at most $\lambda\leq n$.  Thus there are at most $n^4$ distinct tuples.  So we simply define the multiset $L^{\lambda}$ to contain $O(\log^3 n)$ copies of each such tuple, for a total size of $O(n^4 \log^3 n)$.
\end{proof}

This decomposition lemma will allow us to cover \emph{all} demands in our relaxation, even using the nearly-regular assumption in the \SmES relaxation.  More formally, for \LD2S we have the feasibility LP depicted in Figure~\ref{fig:ld2s-lp} with the degree bound $\lambda$ being treated as a constant.

\begin{figure}[htb]
\Hrzline%
\caption{Relaxation $\textbf{LD2S}^\lambda_q(G,\hat E)$}
\label{fig:ld2s-lp}
\hrulefill
\begin{align}
\label{eq:LD2S_LP_SmES} &(z^{u,\tau}_{\emptyset}, (z^{u,\tau}_v)_{v \in V(G_u)}, (z^{u,\tau}_e)_{e \in E(G_u)}) \in \textbf{SmES-LP}^\tau_q(G_u) & \forall u \in V, \forall\tau \in L^\lambda \\
\label{eq:LD2S_LP_decomp} & \sum_{\tau \in L^{\lambda}} z^{u, \tau}_{\emptyset} \leq O(\log^3 \Delta) & \forall u \in V\\
\label{eq:LD2S_LP_z_decomp} & \sum_{\tau \in L^{\lambda}} z^{u,\tau}_{v} + \sum_{\tau \in L^{\lambda}} z^{v,\tau}_u \leq O(\log^3 \Delta) \cdot x_{\aset{u,v}} & \forall \{u,v\} \in E\\
\label{eq:LD2S_LP_degree}& \sum_{v \in \Gamma(u)} x_{\aset{u,v}} \leq \lambda & \forall u \in V \\
\label{eq:LD2S_LP_cover} & x_{\aset{u,v}} + \sum_{w \in \Gamma(u) \cap \Gamma(v)} \sum_{\tau \in L^{\lambda}} z^{w,\tau}_{\{u,v\}} = 1& \forall \{u,v\} \in \hat E \\
\label{eq:LD2S_LP_nonnegative} & 0 \leq x_{\{u,v\}} \leq 1 & \forall \{u,v\} \in E
\end{align}
\Hrzline%
\end{figure}

\paragraph{Remark.}
This LP has size (variables and constraints) $n^{O(q)}$, so can be solved in polynomial time for constant $q$.
Indeed, for each $u \in V$ there are $|L^\lambda|=\tO(n^4)$ different $\smes-lp$ programs.  Each such program has $n^{O(q)}$ variables and constraints, and in addition there are $O(n^2)$ variables of the form $x_{\aset{u,v}}$ and $O(n^2)$ new constraints \eqref{eq:LD2S_LP_decomp}-\eqref{eq:LD2S_LP_nonnegative}.

\begin{lemma}
The feasibility LP~\eqref{eq:LD2S_LP_SmES}-\eqref{eq:LD2S_LP_nonnegative}
is a valid relaxation of \LD2S with degree bound $\lambda$
and demands $\hat E \subseteq E$.
\end{lemma}
\begin{proof}
Let $G = (V,E)$ be a graph, and let $H = (V, E_H)$ be a subgraph of $G$ that is a valid $2$-spanner and has maximum degree $\Delta \leq \lambda$.
We construct an LP solution as follows.
For each $\{u,v\} \in E$, set $x_{\aset{u,v}} = 1$ if $\{u,v\} \in E_H$ and set $x_{\aset{u,v}} = 0$ otherwise.  Note that Constraint~\eqref{eq:LD2S_LP_degree} is satisfied.  For every edge $\{u,v\} \in E \setminus E_H$, there is at least one $2$-path between $u$ and $v$ in $H$, since $H$ is a valid $2$-spanner. Let $w(u,v)$ be the center vertex of such a path, choosing one arbitrarily if there are multiple $2$-paths.  For every vertex $w \in V$, we define the graph $H_w$ to be the subgraph of $G_w$ with vertex set $\{u \in \Gamma(w) : \{u,w\} \in E_H\}$ and all edges $\{u,v\}$ with the property that $w(u,v) = w$ (note that this also implies that $x_{\aset{u,v}} = 0$).  We know from Lemma~\ref{lem:SmES_decomposition} that $H_w$ can be decomposed into $O(\log^3 \Delta)$ nearly-regular graphs with parameters from $L^{\lambda}$ (here $\Delta$ has replaced $n$ because $H_w$ has only $\Delta$ vertices).
Let the multiset $L'_w\subseteq L^{\lambda}$ contain the parameter tuples $\tau$
used in this decomposition, and for $\tau \in L'_w$ let $H_w^{\tau}$ be the graph from the decomposition.  Then for each $\tau \in L'_w$  we set $z^{w,\tau}_{\emptyset} = 1$, and we set $z^{w,\tau}_u = 1$ if and only if $u \in V(H_w^{\tau})$ and set $z^{w, \tau}_{\{u,v\}} = 1$ if and only if $\{u,v\} \in E(H^{\tau}_w)$.  All other variables are $0$.
In particular, for $\tau\notin L'$, even $z^{w,\tau}_\emptyset = 0$.

Let us now prove that all the constraints are satisfied.  Constraint~\eqref{eq:LD2S_LP_SmES} is obviously satisfied, since for any $u \in V$ and $\tau \in L^{\lambda}$ the variables in the \SmES LP are either all $0$ or are integer variables corresponding to $\tau$-nearly regular subgraph.  Constraint~\eqref{eq:LD2S_LP_decomp} is satisfied since $|L'_u| \leq  O(\log^3 \Delta)$, and if $\tau \not\in L'_u$ then $z^{u,\tau}_{\emptyset} = 0$.
Similarly, Constraint~\eqref{eq:LD2S_LP_z_decomp} is satisfied since $|L'_u|$ and $|L'_v|$ are both at most $O(\log^3 \Delta)$ (and thus the left hand side is at most $O(\log^3 \Delta)$), and if any one of the variables on the left is $1$ then the edge $\{u,v\}$ is present in $E_H$ and thus $x_{\aset{u,v}} = 1$.
Finally, Constraint~\eqref{eq:LD2S_LP_cover} is satisfied because for every edge $\{u,v\} \in \hat E$, if $\{u,v\} \in E_H$ then $x_{\aset{u,v}} = 1$ and $\{u,v\}$ does not appear as an edge in any \SmES instance, so the left hand side of the constraint is $1$.  Otherwise $x_{\aset{u,v}} = 0$, in which case $\{u,v\} \in H_{w(u,v)}$ but is not in any other \SmES instance.
Since we decompose $H_{w(u,v)}$, there is exactly one $\tau$ that covers $\{u,v\}$ and for which the corresponding $z^{w,\tau}_{\{u,v\}}=1$,
and thus the left hand side is again $1$.
\end{proof}

\section{Reduction of \LD2S to faithful roundings of  \SmES} \label{sec:reduction}

We now show that if we are given an \algA that is a factor $f$ faithful rounding for $\textbf{SmES-LP}^{\tau}_q(G)$ (even restricted to the case when $z_{\emptyset} = 1$), there is an $\tO(f(\Delta))$-approximation algorithm for \LD2S that uses \algA as a black box.  Lemma~\ref{lem:faithful_from_bipartite} then implies that it is sufficient to be faithful for ${\textbf{Bipartite-SmES-LP}^{\tau}_q(B(G))}$.  Our reduction is given as Algorithm~\ref{alg:LD2S}, which is relatively simple.  It begins with all edges as the demand set $\hat E$, and first solves the LP relaxation for \LD2S with demand set $\hat E$.  It then adds every edge that has $x$ value at least $1/4$.  Then for every \SmES instance in the relaxation it flips a coin, and with probability proportional to $z_{\emptyset}$ (for that instance) scales all $y$ variables for that instance by a factor of $1/z_{\emptyset}$ and then uses the \SmES algorithm on the instance.
After this is completed we update the demands $\hat E$ by removing edges that were successfully covered by this process, and repeat.

\begin{algorithm}[]
\caption{Approximation algorithm for \LD2S}
\label{alg:LD2S}
\SetKwInOut{Input}{Input}\SetKwInOut{Output}{Output}
\Input{Graph $G = (V, E)$, degree bound $\lambda$, factor $f$ faithful rounding \algA for \textbf{SmES-LP}$_q(\cdot)$}
\Output{$2$-spanner $H = (V, E_H)$ of $G$ with maximum degree $\lambda$}
Let $\hat E \leftarrow E$ (unsatisfied demand edges)\\
Let $E_H\leftarrow \emptyset$ (spanner edges)
\\
\While{$\hat E \neq \emptyset$}{
	Compute a solution $\tuple{\vec x,\vec z}$ for
        LP $\textbf{LD2S}^\lambda_q(G,\hat E)$ from Figure~\ref{fig:ld2s-lp}
	\\
	$E_x \leftarrow \{e \in E : x_e \geq 1/4\}$
	\\
	\ForEach{$u \in V, \tau \in L^{\lambda}$}{
		\SetKwBlock{WP}{with probability $z^{u, \tau}_{\emptyset}$}{end}
		\WP{\label{alg2:prob_test}
			$H_{u, \tau} \leftarrow \mathcal A(G_u,\{z^{u,\tau}_b / z^{u,\tau}_{\emptyset}\}_{b \in V(G_u) \cup E(G_u)} )$ \label{alg2:applyA}
			\tcp*{output of \SmES rounding}
			
			$E_{u, \tau} \leftarrow \big\{\{u,v\} \in E : v \in V(H_{u, \tau})\big\}$
		}
		\SetKwBlock{ElseProb}{else (with remaining probability)}{end}
		\ElseProb{
			$E_{u, \tau} \leftarrow \emptyset$
		}
	}
	\tcp{Add all edges found in above rounding}
	$E_H \leftarrow E_H \cup E_x \cup\left(\bigcup_{u \in V} \bigcup_{\tau \in L^{\lambda}} E_{u,\tau} \right)$
	\\
	\tcp{Remove satisfied demands}
        $\hat E \leftarrow \hat E \setminus\left(E_H \cup 
        \big\{ \aset{u,v} : \exists w\in V \textrm{ s.t. both } \{u,w\},\{w,v\} \in E_H 
        \big\} \right)$
}
\end{algorithm}

\begin{theorem} \label{thm:LD2StoSmES}
Let \algA be a factor $f$ faithful rounding for $\textbf{SmES-LP}^{\tau}_q(G)$ with $z_{\emptyset} = 1$.
Then there is a (randomized) $\tO(f(\Delta))$-approximation for \LD2S.
\end{theorem}
\begin{proof}
We may assume that our \LD2S algorithm has a valid guess for $\lambda=\OPT$ by simply trying all $\Delta$ relevant values %
and reporting the best solution.
In this case, the linear program $\textbf{LD2S}^\lambda_q(G,\hat E)$
is guaranteed to have a feasible solution for any $\hat E\subseteq E$.
We now use Algorithm~\ref{alg:LD2S} with this value of $\lambda$.
The proof has two parts: first, we show that in each iteration of the main loop, \whp the set of edges added to $E_H$ (namely $E_x \cup\left(\bigcup_{u \in V} \bigcup_{\tau \in L^{\lambda}} E_{u,\tau} \right)$) forms a subgraph with maximum degree at most $\tO(f(\Delta))\cdot \lambda$.  Second, we show that \whp there are only $\tO(1)$ iterations of the main loop.
Clearly these together prove the theorem.

We begin by analyzing the cost (i.e.~maximum degree) of a single iteration of the main loop.  Let $u \in V$ be an arbitrary vertex; we will show that at most $\tO(f(\Delta))\cdot \OPT$ edges were added incident to it.  There are three types of edges incident to $u$ that are added: these in $E_x$, those in $E_{u,\tau}$ for some $\tau$, and those in $E_{v ,\tau}$ for some neighbor $v$ of $u$ and $\tau \in L^{\lambda}$.  For the first type, the number of edges in $E_x$ incident on $u$ is at most $|\{v \in \Gamma(u) : x_{\aset{u,v}} \geq 1/4\}| \leq \sum_{v \in \Gamma(u)} 4x_{\aset{u,v}} \leq 4\lambda$, where we used Constraint~\eqref{eq:LD2S_LP_degree} to bound the sum.

For the second type of edges, let $L_u$ be the multiset of parameters from $L^{\lambda}$ on which we actually used \algA on $u,\tau$ (i.e.~lines $7$ and $8$ were executed).  Then by Constraint~\eqref{eq:LD2S_LP_decomp} the expected size of $L_u$ is at most $\sum_{\tau \in L^{\lambda}} z^{u, \tau}_{\emptyset} \leq O(\log^{3} n)$, and since the random coins of \algA are independent for each $\tau$, a simple Chernoff bound implies that this holds with high probability (say, probability more than $1-1/n^4$).  For each $\tau \in L_u$, the size of $E_{u,\tau}$ is equal to $|V(H_{u, \tau})|$.  We know from part~\ref{def:absoluteVE-bounds} of the definition of faithful rounding that the size of this set is at most $f(\Delta) \cdot \sum_{v \in \Gamma(u)} z^{u, \tau}_{v}$, and now Constraint~\eqref{eq:smes1} from the \SmES LP (with $T = \emptyset$) together with~\eqref{eq:smes_zu} implies that this is at most $f(\Delta)(k_0 + k_1) \leq 2\lambda f(\Delta)$.  Thus the size of $\bigcup_{\tau \in L^{\lambda}} E_{u,\tau}$ is with high probability at most $O(\log^{3} n) \lambda f(\Delta) \leq \tO(f(\Delta)) \cdot \OPT$.

For the third type of edges, let $v \in \Gamma(u)$ and $\tau \in L^{\lambda}$.  Then the probability that $\{u,v\} \in E_{v, \tau}$ is at most $z^{v,\tau}_{\emptyset} (1/z^{v, \tau}_{\emptyset}) z^{v, \tau}_{u} = z^{v, \tau}_{u}$, where the first $z^{v,\tau}_{\emptyset}$ factor is from the probability of applying \algA to this \SmES instance, the $(1/z^{v, \tau}_{\emptyset})$ factor is from the scaling of the variables, and the $z^{v, \tau}_{u}$ factor is from the definition of faithful rounding.
Now~\eqref{eq:LD2S_LP_z_decomp} and a simple union bound imply that the probability that $\{u,v\} \in \bigcup_{\tau \in L^{\lambda}}E_{v, \tau}$ is at most $\sum_{\tau \in L^{\lambda}} z^{v,\tau}_u \leq O(\log^3 \Delta) \cdot x_{\{u,v\}}$.
This is independent for each $v \in \Gamma(u)$, so by a Chernoff bound we get that with high probability the number of type $3$ edges is at most $O(\log^{3} n) \sum_{v \in \Gamma(u)} x_{\aset{u,v}} \leq O(\log^{3} n) \cdot \lambda \leq \tO(1) \cdot \OPT$, where we used Constraint~\eqref{eq:LD2S_LP_degree} to bound the sum.

Now it just remains to show that only $\tO(1)$ iterations of the main loop are necessary.  We will show that in every iteration at least a $c' = \tOmega(1)$ fraction of the remaining demands $\hat E$ are satisfied in expectation, or equivalently that in expectation the number of remaining unsatisfied demands is at most a $(1-c')$ fraction of the previous number of demands.
To see that this is sufficient, note that by Markov's inequality with probability at most $1- c'/2$ the number of remaining demands is at least a $\frac{1}{1-c'/2} (1-c') = 1 - \frac{c'/2}{1-c'/2}$ fraction of what it was.  Equivalently, with probability at least $c'/2$ at least a $\frac{c'/2}{1-c'/2}$ fraction of demands are covered.  Thus the probability that this does not happen after $(8/c') \ln n = \tO(1)$ iterations is at most $(1-c'/2)^{(8/c') \ln n} \leq 1/n^4$.  So with high probability, after $\tO(1)$ rounds the number of unsatisfied demands is at most $1 - \frac{c'/2}{1-c'/2} \leq 1-c'/2$ of what it was.  Now if this happens $(2/c')\ln n = \tO(1)$ the number of remaining demands is at most $|E|(1-c')^{(2/c') \ln n} < 1$, so the algorithm terminates.  Thus with high probability the number of iterations is at most $(8/c') \ln n \cdot (2/c')\ln n = \tO(1)$, as required.

So now we just need to bound the expected number of demands satisfied in a single iteration, and show that this is $\tOmega(|\hat E|)$.  We break this into two cases.
If $\sum_{\{u,v\} \in \hat E} x_{\aset{u,v}} \geq |\hat E|/2$, then a simple averaging argument shows that at least a $1/3$ fraction
of the edges $\{u,v\}\in \hat E$ are included in $E_x$,
and since each one clearly covers itself as a demand
a total of $\Omega(|\hat E|)$ demands are covered.

In the second case we have that $\sum_{\{u,v\} \in \hat E} x_{\aset{u,v}} < |\hat E|/2$.  We can sum Constraint~\eqref{eq:LD2S_LP_cover} over all demands in $\hat E$, giving us
\begin{equation}\label{eq:coverage_case2}
\sum_{\{u,v\} \in \hat E} \sum_{w \in \Gamma(u) \cap \Gamma(v)} \sum_{\tau \in L^{\lambda}} z^{w,\tau}_{\{u,v\}} = |\hat E| - \sum_{\{u,v\} \in \hat E} x_{\aset{u,v}} > |\hat E|/2
\end{equation}

For each $\{u,v\} \in \hat E$ and $w \in \Gamma(u) \cap \Gamma(v)$ and $\tau \in L^{\lambda}$, let $p^{w,\tau}_{\{u,v\}}$ be the probability that the \SmES rounding of $\textbf{SmES-LP}^\tau_q(G_w)$ %
 covers $\{u,v\}$.
Note that $p^{w,\tau}_{\{u,v\}}$ is just $z^{w,\tau}_{\emptyset}$ times the probability that $\{u,v\}$ is covered by the \SmES rounding of $\textbf{SmES-LP}^\tau_q(G_w)$ assuming that we actually perform this rounding. The expected total number of times edges in $\hat E$ are covered in this phase (with repetitions), can be written as follows:
\begin{align*}
\sum_{\{u,v\} \in \hat E} \sum_{w \in \Gamma(u) \cap \Gamma(v)} \sum_{\tau \in L^{\lambda}} p^{w,\tau}_{\{u,v\}} &=
\sum_{w \in V} \sum_{\tau \in L^{\lambda}} \sum_{\{u,v\} \in \hat E : u,v \in \Gamma(w)} p^{w,\tau}_{\{u,v\}} \\
& \geq  \sum_{w \in V} \sum_{\tau \in L^{\lambda}} z^{w,\tau}_{\emptyset}  \cdot \tOmega\Big(\sum_{\{u,v\} \in \hat E : u,v \in \Gamma(w)} z^{w,\tau}_{\{u,v\}} / z^{w,\tau}_{\emptyset}\Big) \\
& = \tOmega\Big(\sum_{\{u,v\} \in \hat E} \sum_{w \in \Gamma(u) \cap \Gamma(v)} \sum_{\tau \in L^{\lambda}} z^{w,\tau}_{\{u,v\}}\Big) \\
& \geq  \tOmega(|\hat E|), &\text{by~\eqref{eq:coverage_case2}}%
\end{align*}
where %
the first inequality is from the definition of $p^w_{\{u,v\}}$ and part $4$ of Definition~\ref{def:faithful}.  %

Furthermore, we know from the second part of the definition of faithful rounding that $p^{w,\tau}_{\{u,v\}} \leq  z^{w,\tau}_{\emptyset} (z^{w,\tau}_{\{u,v\}} / z^{w,\tau}_{\emptyset}) = z^{w,\tau}_{\{u,v\}}$, so by~\eqref{eq:LD2S_LP_cover} we have $$\sum_{w \in \Gamma(u) \cap \Gamma(v)} \sum_{\tau \in L^{\lambda}} p^{w,\tau}_{\{u,v\}} \leq \sum_{w \in \Gamma(u) \cap \Gamma(v)} \sum_{\tau \in L^{\lambda}} z^{w,\tau}_{\{u,v\}} \leq 1.$$
We can then deduce that the probability that we cover demand $\{u,v\} \in \hat E$
(in a single iteration)
is at least $\tfrac12 \sum_{w \in \Gamma(u) \cap \Gamma(v)} \sum_{\tau \in L^{\lambda}} p^{w,\tau}_{\{u,v\}}$,
by simply using the following well-known argument:
if $t$ (pairwise) independent events occur with probabilities
$q_1,\ldots,q_t$ that sum up to $\sum_{i=1}^t q_i \leq 1$,
then by Bonferroni inequality,
the probability that at least one of these events occurs is at least
\begin{align}
  \label{eq:ByBonferroni}
  \sum_i q_i - \sum_{i<j} q_i q_j
 = \sum_i q_i - \tfrac12\sum_{i\neq j} q_i q_j
  \geq \tfrac12 \sum_i q_i.
\end{align}
We thus obtain that the expected number of demands covered
(in a single iteration) is at least
$$\tfrac12 \sum_{\{u,v\} \in \hat E} \sum_{w \in \Gamma(u) \cap \Gamma(v)} \sum_{\tau \in L^{\lambda}} p^{w,\tau}_{\{u,v\}}
\geq \tOmega(|\hat E|),$$
which completes the proof.
\end{proof}

\section{Structure of faithful rounding algorithm for \SmES}\label{sec:SmES-algorithm}

We devote the remaining sections to our rounding algorithm for \SmES.  Given an $n$-vertex bipartite graph $G=(V_0,V_1, E)$, and a solution to the LP relaxation\footnote{
We will assume in the remaining sections that the LP solution is normalized. That is, we assume that $y_{\emptyset}=1$ (since otherwise, we normalize all variables by defining $y'_T=y_T/y_\emptyset$).} $\textbf{Bipartite-SmES-LP}_q(G, k_0, k_1, d_0, d_1)$, our algorithm will give a factor-$f$ faithful rounding (see Definition~\ref{def:faithful}), for some factor $f=f(n,k_0,k_1,d_0,q)$ which we define in Section~\ref{sec:params}. Given the parameter $q$ which determines the size of our LP and hence the running time (as noted earlier, these are bounded by $n^{O(q)}$), our approximation factor $f$ will be at most $\tO(n^{3-2\sqrt{2}+O(1/q)})$.

At a high level, our algorithm finds a carefully chosen collection of (not necessarily induced) constant-size subgraphs of $G$ in the form of caterpillars, and then samples vertices and edges from the union of these caterpillars according to a very specific distribution.\footnote{As described previously, when used to approximate \LD2S the sampled edges are used only in the analysis while the sampled vertices are used to buy spanner edges}

Finally, in Appendix~\ref{sec:small-degrees}, we give a much simpler faithful rounding algorithm which achieves approximation factor $d_0$. If it happens that $d_0\leq f$, we will run the algorithm described in the appendix, and otherwise we will run our main algorithm. Thus, in the remaining sections, we will assume throughout that $f\leq d_0$. Recall that by convention we assume that $k_0 \geq k_1$ and thus $d_1 \geq \Omega(d_0 / \log n)$, so $f \leq d_0 \leq d_1 \log n$.

\subsection{A simplified goal}

Let us identify some simplified conditions which are sufficient in order to achieve a factor $f$ faithful rounding. We start by weakening the definition of faithful rounding:

\begin{definition}\label{def:weakly-faithful}
A randomized rounding \algA is a \emph{factor-$f$ weakly faithful rounding} if there is some deterministically chosen $\varphi\leq 1$ (possibly $o(1)$), such that when given a solution $(y_a)_{a\in V,E}$ to an LP relaxation for \SmES on a graph
$G=(V,E)$, the algorithm produces a random (not necessarily induced) non-empty subgraph $H^*=(V^*,E^*)$ such that
\begin{enumerate}
\item $\Pr[v \in V^*] \leq \varphi f \cdot y_v$ for all $v \in V$,
\item $\Pr[\{u,v\}\in E^*] \leq \varphi y_{\{u,v\}}$ for all $\{u,v\} \in E$, %
\item \label{def:weak-absoluteVE-bounds}
$|V^*| \leq \varphi f\cdot \sum_{v\in V}y_{v}$ \, (with probability $1$), and
\item \label{def:weak-expectedE-LB}
$\E[|E^*|] \geq \tOmega\left(\varphi\cdot\sum_{\{u,v\} \in E} y_{\{u,v\}}\right)$.
\end{enumerate}
\end{definition}

It turns out that this is equivalent to our original notion:

\begin{lemma}\label{lem:weakly-faithful} For every graph $G$ and LP solution $(y_a)$ as above, if there is an algorithm that achieves a weakly-faithful factor-$f$ rounding, then there is also an algorithm that achieves a faithful factor-$\widetilde{O}(f)$ rounding.
\end{lemma}
\begin{proof} Run the weakly-faithful algorithm $1/\varphi$ times, and take the union of all subgraphs returned.  A trivial union bound implies that this algorithm satisfies the first three parts of the definition of faithful with factor $f$.  To see that it satisfies the fourth part (that the expected number of edges is at least $\tOmega\left(\sum_{\{u,v\} \in E} y_{\{u,v\}}\right)$), let $p_{\{u,v\}}$ denote the probability that the edge $\{u,v\}$ is included in a single iteration. %
Since $p_{\{u,v\}} \leq \varphi y_{\{u,v\}} \leq \varphi$, the probability that $\{u,v\}$ is included in at least one iteration is $1-(1-p_{\{u,v\}})^{1/\varphi} \geq (1-e^{-1})p_{\{u,v\}}/\varphi$. %
So by linearity of expectations the expected number of edges covered in the union is at least $\frac{1-e^{-1}}{\varphi} \sum_{\{u,v\} \in E} p_{\{u,v\}} \geq \tOmega\left(\sum_{\{u,v\} \in E} y_{\{u,v\}}\right)$ as required (where the inequality is from part 4 of the definition of weakly faithful).
\end{proof}

In Section~\ref{sec:SmES_uniformity}, we will show that our rounding algorithm is faithful, with approximation factor determined by the set cardinalities (of subgraph vertices, and of subgraph edges). In other words, we show that it suffices to prove an approximation guarantee in the usual sense, and faithfulness will follow directly. Anticipating this shift in focus, let us define an even weaker notion of faithfulness (which will not be sufficient in general) in terms of set cardinalities:

\begin{definition}
A randomized rounding \algA is a \emph{skewed proportional rounding with parameters $(k'_0,k'_1,m')$} if when given a solution $(y_a)_{a\in V,E}$ to an LP relaxation for \SmES with parameters $(k_0,k_1,d_0,d_1)$ on a bipartite graph
$G=(V_0,V_1,E)$, the algorithm produces a random (not necessarily induced) subgraph $H^*=(V^*_0,V^*_1,E^*)$ such that
\begin{enumerate}
\item $\Pr[u \in V_0^*] \leq \frac{k'_0}{k_0}\cdot y_u$ for all $u \in V_0$,
\item $\Pr[v \in V_1^*] \leq \frac{k'_1}{k_1}\cdot y_v$ for all $v \in V_1$,
\item $\Pr[e\in E^*] \leq \frac{m'}{m}\cdot y_e$ for all $e \in E$, %
\item \label{def:skew-absoluteVE-bounds}
$|V^*_0| \leq k'_0$ and $|V^*_1|\leq k'_1$\, (with probability $1$), and
\item \label{def:skew-expectedE-LB}
$\E[|E^*|] \geq \tOmega\left(m'\right)$.
\end{enumerate}
\end{definition}

As defined, a skewed proportional rounding does not give us any guarantee. The reason for this is that the inflation factors relative to the LP values can be different for nodes in $V_0$ and $V_1$. Namely, we might have $\frac{k'_0}{k_0}\neq\frac{k'_1}{k_1}$.  Let us determine some sufficient condition on the parameters $k'_0,k'_1, m'$ which will guarantee a faithful rounding. Suppose initially, we are given an algorithm which gives a skewed proportional rounding where $k'_0\gg k_0f$ and $k'_1 \gg k_1f$. We could prune the subgraph, by taking a uniformly chosen subset $V'_0\subseteq V^*_0$ of size $(k_0f/k'_0)|V^*_0|$, and a uniformly chosen subset $V'_1\subseteq V^*_1$ of size $(k_1f/k'_1)|V^*_1|$. %
 The expected number of edges in the remaining subgraph (with vertices $(V'_0,V'_1)$) is $$m'\cdot\frac{k_0f}{k'_0}\cdot\frac{k_1f}{k'_1}.$$
If this is at least $\tilde\Omega(m)$, then it is easy to see that we have a faithful rounding. %
 However, in the general case, this last condition may not be sufficient, since we might have $k'_0 \ll k_0f$ or $k'_1 \ll k_1f$ (or both), and then the pruning step does not apply. However, if we also have sufficient average degree in $H^*$ (without pruning), that is, $m'/k'_0\geq d_0/f$ and $m'/k'_1\geq d_1/f$, then this is sufficient. %

\begin{lemma}\label{lem:simplified-goal} %
Suppose we have a skewed proportional rounding with parameters $(k'_0,k'_1,m')$ which satisfy the following three conditions:
  \begin{enumerate}
    \item $m'/k'_0\geq d_0/f$, \label{item:simplified-k_0}
    \item $m'/k'_1 \geq d_1/f$, \label{item:simplified-k_1}
    \item $m'\cdot\frac{k_0f}{k'_0}\cdot\frac{k_1f}{k'_1} \geq m=\tilde\Theta(k_0d_0)$. \label{item:simplified-ratios}
  \end{enumerate}
Then we can also get a faithful factor-$f$ rounding.
\end{lemma}
\begin{proof}
  Let us consider two cases.

Case 1: $k'_0\geq k_0f$ and $k'_1\geq k_1f$. In this case, as explained above, we can sample random subsets of $V^*_0$ and $V^*_1$ of the correct size in order to get the correct number of vertices on each side (Part 3 of Definition~\ref{def:faithful}), and then by condition~\eqref{item:simplified-ratios}, the number of edges is sufficient (Part 4 of Definition~\ref{def:faithful}). Note also, that the individual vertex probabilities are appropriately rescaled (giving Part 1 of Definition~\ref{def:faithful}) by the sampling. Finally, let $\rho=\frac{k_0f}{k'_0}\cdot\frac{k_1f}{k'_1}$ be the subsampling probability of edges in $E^*$ (i.e.\ the probability that an edge in $E^*$ is retained). If $m'\rho=m$, this implies both Parts 2 and 4 of Definition~\ref{def:faithful}. Otherwise, subsample every remaining edge with probability $m/(m'\rho)$, which then guarantees both parts.

Case 2: $k'_0\leq k_0f$ or $k'_1\leq k_1f$. Without loss of generality, suppose $k'_0/k_0 \leq k'_1/k_1$. In this case, let $V'_0=V^*_0$, and take a uniformly chosen subset $V'_1\subseteq V^*_1$ of size $(k'_0k_1/(k_0k'_1))|V^*_1|$. %
Thus the vertex sets have the correct cardinalities for weakly faithful rounding (with scaling factor $\varphi=k'_0/(k_0f)$). Since we only pruned on one side ($V^*_1$), the expected average degree in $V^*_1$ has not changed, which is all we need to show. More formally, the expected number of edges is at least $m'k'_0k_1/(k_0k'_1)\geq d_1k'_0k_1/(fk_0) = mk'_0/(fk_0) = \varphi\cdot m$, where the inequality follows from condition~\eqref{item:simplified-k_1}. Parts 1 and 2 of Definition~\ref{def:weakly-faithful} follow by similar arguments to Case 1.
\end{proof}

\begin{corollary}\label{cor:simplified-goal1} Given a skewed proportional rounding with parameters $(k'_0,k'_1,m')$, then denoting by $d'_b=m'/k'_b$ the average degree in $S^*_b$ (for $b=0,1$), the following three conditions suffice for a faithful rounding:
  \begin{enumerate}
    \item $d'_0 \geq d_0/f$,
    \item $d'_1 \geq d_1/f$,
    \item $\frac{d'_b}{k'_{1-b}} \geq \frac{d_0}{k_1f^2}$ (for either $b=0$ or $b=1$).
  \end{enumerate}
\end{corollary}

\begin{corollary}\label{cor:simplified-goal2} Given a skewed proportional rounding with parameters $(k'_0,k'_1,m')$, then denoting by $d'_b=m'/k'_b$ the average degree in $S^*_b$ (for $b=0,1$), the following three conditions suffice for a faithful rounding:
  \begin{enumerate}
    \item $d'_0 \geq d_0/f$,
    \item $d'_1 \geq d_1/f$,
    \item $k'_b\leq k_bf$ (for either $b=0$ or $b=1$).
  \end{enumerate}
\end{corollary}
\begin{proof} Follows by substituting the lower bound $d_{1-b}/f$ for $m'/k'_{1-b}$ (Condition 1 or 2) in the original condition 3, and using $k_0d_0 = \tilde\Theta(k_1d_1)$.
\end{proof}

\subsection{Algorithm description}\label{sec:algorithm-description}

Our algorithm is a variation of the non-faithful \DkS algorithm of~\cite{BCCFV10}, and as with their algorithm is fundamentally concerned with \emph{caterpillar graphs}.  A caterpillar is a tree consisting of a main path (the \emph{backbone}), with various disjoint paths (\emph{hairs}) attached by one of their endpoints to the backbone. We concentrate on caterpillars with hair-length $1$ (single edges). Let us state the following definition from~\cite{BCCFV10}, which forms the basic template for the algorithm on graphs with maximum degree $n^{r/s}$:

\begin{definition}\label{def:rs-caterpillar} An $(r,s)${\em-caterpillar} is a tree constructed inductively as follows: Begin with a single vertex as the leftmost node in the backbone. For $s$ steps, do the following: at step $t$, if the interval $((t-1)r/s, tr/s)$ contains an integer, add a hair of length $1$ to the rightmost vertex in the backbone; otherwise, add an edge to the backbone (increasing its length by $1$).
\end{definition}

The algorithm will follow along the above iterative construction of an $(r,s)$-caterpillar. At step $t$ of the construction, we will consider a union of $t$-edge prefixes of $(r,s)$-caterpillars in the graph $G$. These will be chosen by an elaborate pruning process described in Section~\ref{sec:bucketing}, which will ensure certain uniformity properties while maintaining a large LP weight associated with these caterpillars. As part of the pruning process, at step $t$, we will consider a certain bipartite subgraph $G_t=(S_t,W_t,E_t)$, whose edges will come from the union of edges added in step $t$ to the above caterpillars (the $t$'th edge in the construction of each).

Furthermore, we will consider the tuple of leaves added before the $t$'th edge in an $(r,s)$-caterpillar. Our bucketing will isolate a set of such leaf tuples $\widetilde{L}_t$, and for every leaf tuple $\lambda\in\widetilde{L}_t$, we will associate a certain subgraph $H^\lambda_t$ of $G_t$, whose edges come from the $t$'th edges in caterpillar prefixes whose initial leaves (up to the $t$th edge) correspond to $\lambda$. We will write the bipartite graph $H^\lambda_t$ as $H^\lambda_t=(U(H^\lambda_t),W(H^\lambda_t),E(H^\lambda_t))$. Note that since we only consider constant size objects, the set $\widetilde{L}_t$ contains at most a polynomial number of tuples.

We are now ready to present our algorithm, in Figure~\ref{fig:Faithful_SmES}, which will give a skewed-proportional rounding satisfying the conditions of Lemma~\ref{lem:simplified-goal} (though under certain conditions the rounding may also be weakly faithful, or even faithful). We remind the reader that this algorithm is intended only for parameters satisfying $f\leq d_0$ (otherwise, we use the much simpler algorithm in Appendix~\ref{sec:small-degrees}). We defer the details of how the above sets are defined to Section~\ref{sec:bucketing}. For now, we only note that they can be computed in parallel to the following algorithm, as needed. Our approximation guarantee $f$ and parameters $r,s$ (all of which depend on $n$, $k_0$, $k_1$, $d_0$, $d_1$, and $q$) will be defined in Section~\ref{sec:params}.

\begin{figure}[t]
\begin{center}\fbox{
\begin{minipage}{15.5 cm}
\textbf{Faithful-S$m$ES$(G,\{y_I\})$}

\noindent Input: A graph $G$ with LP solution $\{y_I\}$ to $\textbf{Bipartite-SmES-LP}_q(G, k_0, k_1, d_0, d_1)$.\\

\noindent For steps $t=1,\ldots,s$:
\begin{enumerate}
    \item\label{step:large-degree} If the maximum degree in $G_t$ is at least $\widetilde\Omega\left(\frac{nd_0}{k_1f^2}\right)$, do the following, and halt:
    \begin{itemize}
      \item If $S_t\subseteq V_0$, choose a subset of $k_0f$ vertices in $S_t$ uniformly at random, and choose $k_1f$ vertices in $W_t$ uniformly at random. Otherwise, choose $k_1f$ vertices in $S_t$ and $k_0f$ in $W_t$ uniformly at random. Output these vertices.
      \item For every edge $e=(u,v)$ whose endpoints are chosen in the previous step, choose edge $e$ with probability $\frac{y_{\{e\}}}{(y_{\{u\}})y_{\{v\}}f^2}$. Output these edges.
    \end{itemize}
    \item Let $m'=\expec_{\lambda\in_R\widetilde{L}_t}[|E(H^\lambda_t)|]$. If $S_t\subseteq V_0$, let $k'_0=\expec_{\lambda\in_R\widetilde{L}_t}[|U(H^\lambda_t)|]$ and $k'_1=\expec_{\lambda\in_R\widetilde{L}_t}[|W(H^\lambda_t)|]$. Otherwise, let $k'_0=\expec_{\lambda\in_R\widetilde{L}_t}[|W(H^\lambda_t)|]$ and $k'_1=\expec_{\lambda\in_R\widetilde{L}_t}[|U(H^\lambda_t)|]$.
    \item\label{step:default} If the parameters $m',k'_0,k'_1$ satisfy the conditions of Lemma~\ref{lem:simplified-goal}, choose $\lambda\in\widetilde{L}_t$ uniformly at random, output $H^\lambda_t$, and halt.
\end{enumerate}
\end{minipage}
}%
\end{center}\caption{Algorithm \bf{Faithful-S$m$ES}}\label{fig:Faithful_SmES}
\end{figure}

\section{Details of our rounding algorithm for bipartite \SmES}\label{sec:SmES-details}

In this section, we give details of our bucketing procedure which defines the various subgraphs used in algorithm \textbf{Faithful-S$m$ES}. Moreover, we specify the parameters $r,s$ and our approximation ratio $f$ which define the specific caterpillar structure on which our algorithm is based.

\subsection{Parametrization and maximum degree bound} \label{sec:params}

Given parameters $n,k_0,k_1,d_0,d_1,q$, we need to define our intended approximation ratio $f$, and the choice of $r$ and $s$ in the $(r,s)$-caterpillar which will determine the structure of our graph, and will also give an effective bound on the maximum degree. In this section, we define these parameters. As before, let $q$ to be the (bounded) integer parameter which corresponds to the level of the Sherali-Adams relaxation we use. The running time both for solving the LP and for our rounding will be $n^{O(q)}$ while %
 the approximation guarantee will be $\widetilde{O}(n^{3-2\sqrt2 + O(1/q)})$.

Similarly to~\cite{BCCFV10}, we would like to choose $r$ and $s$ so that the maximum degree will be roughly $n^\alpha=n^{r/s}$ for $s\leq q$. %
 In Appendix~\ref{sec:max-deg-bound}, we will show that Step~\ref{step:large-degree} of the algorithm gives a faithful rounding whenever the maximum degree is greater than $D=\frac{n}{k_1 f}\cdot\frac{d_0}{f}$. Thus we would to define $\alpha=r/s$ so that 
 \begin{equation}\label{eq:approx-D-bound}
n^\alpha\approx \frac{n}{k_1 f}\cdot\frac{d_0}{f}.
\end{equation}
 We will also need our approximation factor $f=f(n,k_0,k_1,d_0,d_1,q)$ to satisfy the following crucial property:
\begin{equation}\label{eq:f-property}
f=\left(\frac{k_1f}{d_0}\right)^\alpha.
\end{equation}
 
 Note that if we had equality in both~\eqref{eq:approx-D-bound} and~\eqref{eq:f-property}, then together, these implicitly define both $\alpha$ and $f$. However, recall that we require $\alpha=r/s$ to be rational (specifically, we will want $\alpha$ to be a multiple of $1/q$), which most likely would not occur if we had strict equality in~\eqref{eq:approx-D-bound}.  To fix this, we simply round the implied value of $\alpha$ up to the nearest multiple of $1/q$. It is straightforward to check that, if we let $\gamma=\log_n(k_1/d_0)$, then the new value of $\alpha$ will be

$$\alpha={\textstyle\frac1q}\left\lceil q(1+\gamma/2-\sqrt{2\gamma+\gamma^2/4})\right\rceil.$$

This specific closed-form expression for $\alpha$ is never used in our analysis. It is only important in that it guarantees~\eqref{eq:approx-D-bound}, implicitly defines $f$ through~\eqref{eq:f-property}, and gives the following guarantee, which is the precise formulation of~\eqref{eq:approx-D-bound}:
\begin{prop}\label{prop:final-param-ineq} For $f$, $D$ and $\alpha$ defined as above, we have $$D={\textstyle\frac{n}{k_1 f}\cdot\frac{d_0}{f}}=n^{\alpha-O(1/q)}.$$
\end{prop}
\noindent{}This implies that
\begin{align*}
 f&=n^{1-\alpha+O(1/q)}d_0/(k_1f)\\
 &=n^{1-\alpha+O(1/q)}/f^{1/\alpha}.&\text{by~\eqref{eq:f-property}}
 \end{align*}
 Thus, in the worst case, we have the following upper bound on our approximation factor $f$:
\begin{corollary}\label{cor:f-final-bound}
For $f$, $D$ and $\alpha$ defined as above,
$$f=n^{(1-\alpha+O(1/q))\alpha/(1+\alpha)}\leq\max_{0\leq\alpha\leq1}n^{(1-\alpha)\alpha/(1+\alpha)+O(1/q)}=n^{3-2\sqrt2+O(1/q)}\approx n^{0.172+O(1/q)}.$$
\end{corollary}

Finally, we choose integers $r,s$ such that $r$ and $s$ are co-prime and $r/s=\alpha$.

\subsection{LP rounding and bucketing} \label{sec:bucketing}

The fundamental difficulty in adapting LP-hierarchy-based algorithms such as the one in~\cite{BCCFV10} to work in a faithful randomized rounding setting is that the LP values which are finally used in finding the actual subgraph are values which arise after several rounds of conditioning. These can be markedly different from (or even nearly unrelated to) the original LP values, which we want to round with respect to. We get around this problem by limiting our algorithm to vertices (and edges, and larger structures) which have roughly uniform LP values (for every specific kind of vertex/edge/structure). Thus, even if the values which we use to round are different from the original LP values, they are at least uniformly proportional to them.

In ignoring many vertices/edges/structures in the graph, we have to make sure that we are not losing too much of the information given to us by the LP. Since our algorithm and analysis ultimately rely on examining caterpillars of a certain form (which is determined by our parameters), the goal of our bucketing of caterpillars is to preserve the total LP weight of such caterpillars in our graph, up to a polylogarithmic factor.

Let $K$ be the template $(r,s)$ caterpillar for $r,s$ corresponding to our parameters (see Definition~\ref{def:rs-caterpillar}) %
and denote by $K_t$ the $(t-1)$-edge ($t$-vertex) prefix of this caterpillar (that is, the caterpillar constructed inductively as before, for the first $t-1$ steps). We will now define a procedure which buckets instances in $G$ of these caterpillar prefixes. %

Before describing the procedure, let us introduce one more piece of notation. For a set of caterpillars $B$ (a bucket), for every tuple of leaves $\lambda$ (not including the rightmost backbone vertex, which may be also be a leaf), denote by $T^\lambda(B)$ the set of caterpillars in $B$ which have leaves $\lambda$ (in the same order). In addition, for node $u$, and edge $e$, denote by $T^\lambda_u(B)$ and $T^\lambda_e$ the set of caterpillars in $T^\lambda(B)$ which have $u$ (resp.\ $e$) as the final vertex (resp.\ edge). Note that this edge may be either a backbone edge or a hair. An important invariant will be that at the end of every step, the cardinalities of sets $|T^\lambda(B)|$ (resp.\ $|T^\lambda_u(B)|$) will be roughly uniform (up to polylogarithmic factors) for every remaining choice of leaf tuple $\lambda$ (resp.\ leaf tuple $\lambda$ and rightmost backbone vertex $u$).

\begin{itemize}
\item Initialization: By Constraint~\eqref{eq:smes1}, $\sum_{v\in V_0}y_v=k_0$. Bucket vertices in $V_0$ by their LP-value. Then there is some bucket $B_1$ for which $\sum_{v\in V_0}y_v\geq k_0/\log n$. Let us also write $L_1=S_1=B_1$.
\item Step $t$: Suppose $B_t$ is a bucket of instances of caterpillar $K_t$ in the graph, $S_t$ is the set of rightmost (final) backbone nodes in these caterpillars, and $L_t$ is the set of leaf-tuples occurring in these caterpillars (not including the final backbone vertex). Let $b\in\{0,1\}$ be such that $S_t\subseteq V_b$ (this depends only on the parity of backbone edges in $K_t$). Construct a sequence of nested sets of caterpillars in stages as follows:
\begin{itemize}
  \item Consider the set of caterpillars $\widetilde{B}_t$ formed by taking a caterpillar from $B_t$ and adding an edge to the rightmost node. Then up to a logarithmic factor, by~\eqref{LP:d-bound}, for every caterpillar $K\in B_t$ with rightmost endpoint $u$, we have $\sum_{w\in\Gamma(u)}y_{K\cup\{\{u,w\}\}}=d_by_{K}$. In particular, this implies $\sum_{\widetilde{K}\in\widetilde{B}_t}y_{\widetilde{K}}=d_b\cdot\sum_{K\in B_t}y_{K}$.
  \item Bucket the new caterpillars in $\widetilde{B}_t$ so that the following values are uniform (up to a constant factor): the LP value of the newly added vertex, of the newly added edge, of the leaves together with the new vertex %
  (that is, if the caterpillar originally had leaves $\lambda$ and a new vertex $w$ was added, then this is the LP value $y_{\lambda\cup\{w\}}$) and of the entire caterpillar. Then the heaviest bucket preserves the total LP weight of caterpillars in $\widetilde{B}_t$ up to a polylogarithmic factor. Denote this bucket by $\widetilde{B}'_t$.
  \item Next, bucket the caterpillars in $\widetilde{B}'_t$ so that within every bucket, for every tuple of leaves $\lambda$ and edge $(u,w)$ the cardinality $|T^\lambda_{(u,w)}(\widetilde{B}'_t)|$ is roughly uniform. Again, retain the heaviest bucket (which preserves at least a polylogarithmic fraction of LP weight), and denote this bucket by $\widetilde{B}''_t$.
  \item %
Now bucket $\widetilde{B}''_t$ by the cardinality $|\bigcup_uT^{\lambda}_{(u,w)}(\widetilde{B}''_t)|$ (so that this value is roughly uniform over the choice of leaf tuple $\lambda$ and new vertex $w$), and let $\widetilde{B}'''_t$ denote bucket with the largest LP weight.
  \item Finally, prune the remaining set of caterpillars $\widetilde{B}'''_t$ by retaining only those with leaf tuples $\lambda$ which satisfy $$\sum_{K'''\in T^\lambda(\widetilde{B}'''_t)}y_{K'''}\geq\widetilde{\Omega}(d_b)\sum_{K\in T^\lambda(B_t)}y_K.$$ Note that for such $\lambda$, at least a polylogarithmic fraction of the total LP weight of $T^\lambda(\widetilde{B}_t)$ must have also been preserved in each of the previous stages. Denote this set of leaf tuples by $\widetilde{L}_t$, and let $B_{t+1}\subseteq \widetilde{B}'''_t$ be the corresponding set of caterpillars. %
  \item Denote by $W_t$ the set of all newly added vertices in caterpillars in $B_{t+1}$. That is, $$W_t=\bigcup_{\lambda,u}\{w\mid T^{\lambda}_{(u,w)}(B_{t+1})\neq\emptyset\}.$$
\end{itemize}
\item Backbone step: If the last edge in $K_{t+1}$ is a backbone edge, let $L_{t+1}=\widetilde{L}_{t}$, and let $S_{t+1}=W_t$.
\item Hair step: Otherwise, the last edge in $K_{t+1}$ is a hair. In this case, the rightmost backbone vertices in the caterpillars in $\widetilde{B}_t$ still belong to $S_t$ (in each caterpillar that survives the bucketing, it is the same vertex as before the new edge was added). Let $L_{t+1}$ be the set of new leaf tuples of these caterpillars (where each leaf tuple includes the newly added leaf in $W_t$), and let $S_{t+1}\subseteq S_t$ be the remaining set of rightmost backbone vertices. %
\end{itemize}

Note that (by induction), if caterpillar $K_t$ rooted in $V_0$ has $t_0$ edges emanating from side $V_0$ and $t_1$ edges emanating from side $V_1$ (think of edges as being directed away from the initial node), then bucket $B_t$ satisfies %
\begin{equation}\label{eq:lp-weight-preserve}
\sum_{K\in B_t}y_K=\widetilde{\Theta}(k_0d_0^{t_0}d_1^{t_1}),
\end{equation}
 which is precisely the number of such caterpillars in a $(d_0,d_1)$-regular bipartite graph (assuming we allow degeneracies, such as caterpillars intersecting themselves).

The following lemma shows that LP weight is preserved even on substructures of caterpillars, such as individual vertices and edges.

\begin{lemma}\label{lem:regularity} Let $B$ be a set of trees rooted on side $V_0$ and isomorphic to a tree $Q$ which, when rooted on side $V_0$ with edges directed away from the root, has $t_0$ edges emanating from side $V_0$ and $t_1$ edges emanating from side $V_1$. If the set $B$ satisfies Equation~\eqref{eq:lp-weight-preserve}, then
  \begin{itemize}
    \item For any vertex $i$ in $Q$ on side $V_b$ (for some $b\in\{0,1\}$), letting $B_i$ be the set of all copies of vertex $i$ in trees $R\in B$, we have $$\sum_{v\in B_i}y_{\{v\}}=\widetilde{\Theta}(k_b).$$
    \item For any edge $e$ in $Q$, letting $B_e$ be the set of all copies of vertex $e$ in trees $R\in B$, we have $$\sum_{g\in B_e}y_{\{g\}}=\widetilde{\Theta}(k_0d_0).$$
  \end{itemize}
\end{lemma}
\begin{proof}
We give the proof for vertices on side $V_0$. The proof for vertices on side $V_1$ and for edges is similar. Let $s=\sum_{v\in B_i}y_{\{v\}}$. By repeated applications of the Sherali-Adams Constraint~\eqref{LP:d-bound} we have $$\sum_{R\in B} y_R\leq s d_0^{t_0}d_1^{t_1}.$$ By~\eqref{eq:lp-weight-preserve}, this implies that $s=\widetilde{\Omega}(k_0)$. On the other hand, by Constraint~\eqref{eq:smes1}, we have $s\leq k_0$, which completes the proof.
\end{proof}

\begin{remark} We will take advantage of the bucketing on LP values and other values above to abuse notation. For example, since all caterpillars of the form $K_t$ that survive our bucketing will have the same LP value (up to a constant factor), we can be a bit imprecise, and write $y_{K_t}$, with the understanding that we ignore constant factors in our analysis.
\end{remark}

At step $t$, for leaf tuple $\lambda\in\widetilde{L}_t$, let us define the bipartite subgraph $H^\lambda_t=(U(H^\lambda_t),W(H^\lambda_t),E(H^\lambda_t))$ as follows. %
Let
$$
E(H^\lambda_t)=\{(u,w)\mid T^{\lambda}_{(u,w)}(B_{t+1})\neq\emptyset\},
$$
and let $U(H^\lambda_t)$ and $W(H^\lambda_t)$ be the vertex sets incident to these edges. That is, $U(H^\lambda_t)=\{u\in S_t\mid \exists w\in W_t:(u,w)\in E(H^\lambda_t)\}$, and $W(H^\lambda_t)=\{w\in W_t\mid \exists u\in U(H^\lambda_t):(u,w)\in E(H^\lambda_t)\}$. This subgraph (for an appropriate choice of $t$, and distribution over $\lambda$) will form the basis of our rounding algorithm. Later, we will show that the cardinalities of sets involved do not depend on the choice of leaves (see Lemma~\ref{lem:uniform-cardinalities}).

Note that the graphs $H^\lambda_t$ are subgraphs of the graphs $G_t=(S_t,W_t,E_t)$ which were mentioned in Section~\ref{sec:algorithm-description}. The graph $G_t$ is in fact the union (over leaf tuples $\lambda$) of all subgraphs $H^\lambda_t$. The sets $S_t$ and $W_t$ were defined in the above bucketing, while $E_t$ is defined as follows:
$$
E_t=\{(u,w)\mid T_{(u,w)}(B_{t+1})\neq\emptyset\}.
$$

\section{Performance guarantee of our rounding algorithm: faithfulness and approximation}\label{sec:SmES-analysis}

In this section, we show that Algorithm {\bf{Faithful-S$m$ES}} gives a faithful $f$-approximation (using Lemma~\ref{lem:simplified-goal}). In fact, if at any iteration the lower bound on degrees in Step~\ref{step:large-degree} of the algorithm is satisfied, then Step~\ref{step:large-degree} already gives a faithful rounding on its own. Otherwise (assuming an upper bound on degrees), we show that at some iteration $t$, Step~\ref{step:default} of the algorithm gives a skew-proportional rounding satisfying the conditions of Lemma~\ref{lem:simplified-goal}.

\subsection{LP rounding: faithfulness} \label{sec:SmES_uniformity}

Before we prove the approximation guarantee, we need to show that the algorithm gives a skewed-proportional rounding. In fact, this represents the core of our technical contribution. We will show that the sampling procedure suggested in Step~\ref{step:default} of the algorithm always gives a skewed-proportional rounding, regardless of whether the conditions of Lemma~\ref{lem:simplified-goal} are met. The following lemma allows us to reformulate the conditions of skewed-proportional rounding in terms of set cardinalities.

\begin{lemma}\label{lem:uniformity-to-faithful}
Let $\mathcal{A}$ be an algorithm which, for some $t$, outputs a random subgraph $H^*=(V^*_b,V^*_{1-b},E^*)$ of $G_t=(S_t,W_t,E_t)$ (where $b \in \{0,1\}$ is s.t.\ $S_t\subseteq V_b$), where the cardinalities $|V^*_0|,|V^*_1|,|E^*|$ are roughly uniform (over the randomness). %
 If the following conditions hold:
\begin{enumerate}
\item For all $u\in S_t$, $\prob[u\in V^*_b]=\widetilde{O}(\mathbb{E}[|V^*_b|]/|S_t|)$,
\item for all $w\in W_t$, $\prob[w\in V^*_{1-b}]=\widetilde{O}(\mathbb{E}[|V^*_{1-b}|]/|W_t|)$, and
\item for all $e\in E_t$, $\prob[e\in E^*]=\widetilde{O}(\mathbb{E}[|E^*|]/|E_t|)$,
\end{enumerate}
then algorithm $\mathcal{A}$ is a skewed proportional rounding with parameters $(\widetilde{O}(\mathbb E[|V^*_0|]),\widetilde{O}(\mathbb E[|V^*_1|]),\widetilde{\Omega}(\mathbb E[|E^*|]))$.
\end{lemma}
\begin{proof} Follows directly from the definition of skewed proportional rounding, and the fact that, by bucketing and by Lemma~\ref{lem:regularity}, for all $u\in S_t$, $w\in W_t$ and $e\in E_t$ we have $y_u=\widetilde{\Theta}(k_b/|S_t|)$, $y_w=\widetilde{\Theta}(k_{1-b}/|W_t|)$, and $y_e=\widetilde{\Theta}(m/|E_t|)$.
\end{proof}

Our goal is thus to show three simple lemmas corresponding to the above conditions. One each for the probability of a vertex/edge belonging to $U(H^\lambda_t)$, $W(H^\lambda_t)$, and $E(H^\lambda_t)$ (for uniformly chosen $\lambda\in\widetilde{L}_t$). First, let us show that, indeed, the cardinalities of these sets (and others) do not vary by more than a polylogarithmic factor over the choice of leaf-tuple $\lambda$.

\begin{lemma}\label{lem:uniform-cardinalities} For every step $t$ in the above bucketing, the cardinality of both of the following sets (when non-empty) does not vary by more than a polylogarithmic factor over the choice of leaf tuple $\lambda \in L_t$ and rightmost vertex $u$:
$$T^\lambda(B_t), \qquad T^\lambda_u(B_t).$$
\noindent Moreover, the cardinality of each of the following sets also does not vary by more than a polylogarithmic factor over the choice of leaf tuple $\lambda\in\widetilde{L}_t$:
$$U(H^\lambda_t), \qquad W(H^\lambda_t), \qquad E(H^\lambda_t).$$
\end{lemma}
\begin{proof} 
We proceed by induction. For $t=1$, by definition, for every $v\in L_1$ we have $T^v(B_1)=\{v\}$, and $T^v_u(B_1)=\{v\}$ iff $v=u$ (otherwise it is empty).

Let us now assume that the lemma holds for $T^\lambda(B_t)$ and $T^\lambda_u(B_t)$ (over the choice of $\lambda\in L_t$), and show that it holds for  $U(H^\lambda_t)$, $W(H^\lambda_t)$, $E(H^\lambda_t)$ (over the choice of $\lambda\in \widetilde{L}_t$),  and for $T^{\lambda'}(B_{t+1})$ and $T^{\lambda'}_{u'}(B_{t+1})$ (over the choice of $\lambda'\in L_{t+1}$). Let $\lambda\in\widetilde{L}_t$ be a tuple of leaves. 
As noted earlier (see definition of $\widetilde{L}_t$), for any $\lambda\in\widetilde{L}_t$, at least a polylogarithmic fraction of the total LP weight of $T^\lambda(\widetilde{B}_t)$ must have been preserved in each stage of the bucketing. Since by the inductive hypothesis, the cardinality $|T^\lambda_u(B_t)|$ is roughly uniform (up to a polylogarithmic factor), and since LP values of caterpillars in $B_t$ are roughly uniform, every vertex $u$ with non-empty $T^\lambda_u(B_t)$ contributes roughly the same LP weight to the total weight of $T^\lambda(B_t)$, and therefore a $\widetilde{\Omega}(1)$ fraction of these vertices must also survive the bucketing. In particular, this implies that 
\begin{equation}\label{eq:U-cardinality}|U(H^\lambda_t)|=\widetilde{\Theta}(|\{u\in S_t\mid T^\lambda_u(B_t)\neq\emptyset\}|)=\widetilde{\Theta}(|T^\lambda(B_t)|/|T^\lambda_u(B_t)|).\end{equation}
Since by the inductive hypothesis, the two values in the final ratio on do not depend on $\lambda,u_0$ (by more than a polylogarithmic factor), the claim follows for $U(H^\lambda_t)$.

Next, let us show uniformity of the cardinalities $W(H^\lambda_t)$ and $E(H^\lambda_t)$. By construction, we have
$$|T^\lambda(B_{t+1})|=\frac{1}{y_{K_{t+1}}}\sum_{\widetilde{K}\in T^\lambda(B_{t+1})}y_{\widetilde{K}}=\widetilde{\Theta}\left(\frac{d_b}{y_{K_{t+1}}}\right)\sum_{K\in T^\lambda(B_t)}y_K=\widetilde{\Theta}\left(\frac{d_by_{K_t}}{y_{K_{t+1}}}\right)|T^\lambda(B_t)|.$$
Since by construction of $\widetilde{B}''_t$, every edge $(u,w)\in E(H^\lambda_t)$ participates in the same number of caterpillars in $T^\lambda(B_{t+1})$ (up to a constant factor), call it $|T^\lambda_e(B_{t+1})|$, we have that 
\begin{align}
|E(H^\lambda_t)|&=\frac{|T^\lambda(B_{t+1})|}{|T^\lambda_e(B_{t+1})|}\label{eq:E-cardinality}\\
&=\widetilde{\Theta}\left(\frac{d_by_K|T^\lambda(B_t)|}{y_{K_{t+1}}|T^\lambda_e(B_{t+1})|}\right),\nonumber
\end{align} 
which does not depend on $\lambda$, since all the values in the final expression are fixed by bucketing. By a similar argument (by construction of $\widetilde{B}'''_t$), every vertex $w\in W(H^\lambda_t)$ has the same degree in $H^\lambda_t$ (up to a constant factor). Thus, since the number of edges in $H^\lambda_t$ is fixed (up to a polylogarithmic factor), so is $|W(H^\lambda_t)|$.

Finally, we need to show the uniformity of $T^{\lambda'}(B_{t+1})$ and $T^{\lambda'}_{u'}(B_{t+1})$ (over the choice of $\lambda'\in L_{t+1}$ and $u'$ such that $T^{\lambda'}_{u'}(B_{t+1})\neq\emptyset$). Let us consider the two different kinds of steps separately. If $t$ is a backbone step, then by the bucketing that defines $\widetilde{B}'''_t$, the number of caterpillars $|T^\lambda_w(B_{t+1})|$ is fixed up to a constant factor. Moreover, $\{w\in S_{t+1}\mid T^\lambda_w(B_{t+1})\neq\emptyset\}=W(H^\lambda_t)$, and as we've shown, the number of vertices $w$ in the latter set is fixed up to a polylogarithmic factor, therefore the same holds for the total number of caterpillars $|T^\lambda(B_{t+1})|=\sum_{w\in W(H^\lambda_t)}|T^\lambda_w(B_{t+1})|$.

Now, suppose $t$ is a hair step. As we've noted, for every $(u,w)\in E(H^\lambda_t)$, the number of caterpillars $|T^{\lambda}_{(u,w)}(B_{t+1})|$ is fixed up to a constant factor. However, recall that in a hair step we have $T^{\lambda}_{(u,w)}(B_{t+1})=T^{\lambda\cup\{w\}}_u(B_{t+1})$. Thus, it only remains to show the claim for $|T^{\lambda\cup\{w\}}(B_{t+1})|=\sum_{u\in\Gamma_{H^{\lambda}_t}(w)}|T^{\lambda\cup\{w\}}_u(B_{t+1})|$. But here it is clearly sufficient to show that every vertex $w\in W(H^\lambda_t)$ has roughly the same degree in $H^\lambda_t$, which we have already argued.
\end{proof}

Before we state and prove the three main lemmas which we need for the probability bounds in Lemma~\ref{lem:uniformity-to-faithful}, let us introduce the following notation. Let us assume that when caterpillar $K_t$ has its initial vertex at side $V_0$ with edges directed away from the initial vertex, it has $t_0$ outgoing edges from side $V_0$ and $t_1$ outgoing edges from side $V_1$ (thus $t_0+t_1=t-1$).  %
 For concreteness, let us assume that $K_t$ has an even number of backbone edges. That is, when the initial backbone vertex is in $V_0$ so is the rightmost vertex. The other case (when the backbone has odd length) is quite similar. %
 We are now ready to prove the remaining lemmas which will guarantee skewed proportionality.

\begin{lemma}\label{lem:prob-bound-U} For uniformly chosen $\lambda\in\widetilde{L}_t$, for every $u\in S_t$ we have $\prob_\lambda[u\in U(H^\lambda_t)]=\widetilde{O}(|U(H_t)|/|S_t|)$.
\end{lemma}
\begin{proof} %
  As we have argued earlier, we have $\sum_{\lambda\in\widetilde{L}_t}\sum_{K\in T^{\lambda}(B_t)}y_K=\widetilde{\Omega}(k_0d_0^{t_0}d_1^{t_1})$. By bucketing, we can write this more simply as
  \begin{equation}\label{eq:B_t-sum}
    |\widetilde{L}_t||T^\lambda(B_t)|y_{K_t}=\widetilde{\Omega}(k_0d_0^{t_0}d_1^{t_1}).
  \end{equation}
  Now, take any vertex $u\in S_t$. By repeated applications of Constraint~\eqref{LP:d-bound}, we have
$$\sum_{\lambda\in\widetilde{L}_t:U(H^\lambda_t)\ni u}\;\sum_{K\in T^\lambda_u(B_t)}y_K\leq d_0^{t_0}d_1^{t_1}y_u.$$ By Lemma~\ref{lem:uniform-cardinalities}, for any $\lambda\in\widetilde{L}_t$ such that $U(H^\lambda_t)\ni u$, the number of caterpillars this tuple of leaves contributes, $|T^\lambda_u(B_t)|$, is roughly uniform (up to a polylogarithmic factor). Thus, we have $$|\{\lambda\in\widetilde{L}_t\mid U(H^\lambda_t)\ni u\}||T^\lambda_u(B_t)|y_{K_t}\approx\sum_{\lambda\in\widetilde{L}_t:U(H^\lambda_t)\ni u}\sum_{K\in T^\lambda_u(B_t)}y_K\leq d_0^{t_0}d_1^{t_1}y_u.$$
By Constraint~\eqref{eq:smes1} and bucketing, we have $y_u=O(k_0/|S_t|)$, so the above inequality implies
$$|\{\lambda\in\widetilde{L}_t\mid U(H^\lambda_t)\ni u\}||T^\lambda_u(B_t)|y_{K_t}=O(k_0d_0^{t_0}d_1^{t_1}/|S_t|).$$
Combining this with~\eqref{eq:B_t-sum}, we have
\begin{align*}
\prob_{\lambda\in_R\widetilde{L}_t}[u\in U(H^\lambda_t)] &= \frac{|\{\lambda\in\widetilde{L}_t\mid U(H^\lambda_t)\ni u\}|}{|\widetilde{L}_t|}\\
&=\widetilde{O}\left(\frac{k_0d_0^{t_0}d_1^{t_1}/(|S_t||T^\lambda_u(B_t)|y_{K_t})}{k_0d_0^{t_0}d_1^{t_1}/(|T^\lambda(B_t)|y_{K_t})}\right)
\\&=\widetilde{O}\left(\frac{|T^\lambda(B_t)|/|T^\lambda_u(B_t)|}{|S_t|}\right),
\end{align*}
which by~\eqref{eq:U-cardinality} is what we needed to prove.
\end{proof}

\begin{lemma} For uniformly chosen $\lambda\in\widetilde{L}_t$, for every $w\in W_t$ we have $\prob_\lambda[w\in W(H^\lambda_t)]=\widetilde{O}(|W(H_t)|/|W_t|)$.
\end{lemma}
\begin{proof} The proof is similar to that of Lemma~\ref{lem:prob-bound-U}. For $\lambda\in\widetilde{L}_t$, let us denote by $T^{\lambda,w}(B_{t+1})$ the set of caterpillars in $B_{t+1}$ which have leaves $\lambda$ and $w$. That is, if $t$ is a backbone step, then $T^{\lambda,w}(B_{t+1})=T^{\lambda}_w(B_{t+1})$, and otherwise (if $t$ is a hair step), then $T^{\lambda,w}(B_{t+1})=T^{\lambda\cup w}(B_{t+1})$. Noting that in both cases, the cardinality of this set is roughly uniform for all $\lambda\in\widetilde{L}_t$ and $w\in W(H^\lambda_t)$, let us abuse notation and write $|T^{\lambda,w}|$ for the cardinality of such sets. As argued earlier, the degrees of all $w\in W(H^\lambda_t)$ in $H^\lambda_t$ are roughly uniform, and each edge $(u,v)\in E(H^\lambda_t)$ participates in roughly the same number of caterpillars. Thus, it follows that
  \begin{equation}\label{eq:W-cardinality}
    |W(H_t)|=\widetilde{\Theta}(|T^\lambda(B_{t+1})|/|T^{\lambda,w}(B_{t+1})|).
  \end{equation}
Analogously to~\eqref{eq:B_t-sum}, we can also show
  \begin{equation}\label{eq:B_t+1-sum}
    |\widetilde{L}_t||T^\lambda(B_{t+1})|y_{K_{t+1}}=\widetilde{\Omega}(k_0d_0^{t_0+1}d_1^{t_1}).
  \end{equation}
  Now, take any vertex $w\in W_t$. By repeated applications of Constraint~\eqref{LP:d-bound}, we have
$$\sum_{\lambda\in\widetilde{L}_t:W(H^\lambda_t)\ni w}\;\sum_{K\in T^{\lambda,w}(B_t)}y_K\leq k_0d_0^{t_0}d_1^{t_1+1}.$$ As before, for any $\lambda\in\widetilde{L}_t$ such that $W(H^\lambda_t)\ni w$, the number of caterpillars this tuple of leaves contributes, $|T^{\lambda,w}(B_t)|$, is roughly uniform (up to a polylogarithmic factor). Thus, we have $$|\{\lambda\in\widetilde{L}_t\mid W(H^\lambda_t)\ni w\}||T^{\lambda,w}(B_{t+1})|y_{K_{t+1}}\approx\sum_{\lambda\in\widetilde{L}_t:W(H^\lambda_t)\ni w}\;\sum_{K\in T^{\lambda,w}(B_{t+1})}y_K\leq d_0^{t_0}d_1^{t_1+1}y_w.$$
By Constraint~\eqref{eq:smes1} and bucketing, we have $y_w=O(k_1/|W_t|)$, so the above inequality implies
$$|\{\lambda\in\widetilde{L}_t\mid W(H^\lambda_t)\ni w\}||T^{\lambda,w}(B_{t+1})|y_{K_{t+1}}=O(k_1d_0^{t_0}d_1^{t_1+1}/|W_t|).$$
Combining this with~\eqref{eq:B_t+1-sum}, we have
\begin{align*}
\prob[w\in W(H^\lambda_t)] &= \frac{|\{\lambda\in\widetilde{L}_t\mid W(H^\lambda_t)\ni w\}|}{|\widetilde{L}_t|}\\
&=\widetilde{O}\left(\frac{k_1d_0^{t_0}d_1^{t_1+1}/(|W_t||T^{\lambda,w}(B_{t+1})|y_{K_{t+1}})}{k_0d_0^{t_0+1}d_1^{t_1}/(|T^\lambda(B_{t+1})|y_{K_{t+1}})}\right)
\\&=\widetilde{O}\left(\frac{|T^\lambda(B_{t+1})|/|T^{\lambda,w}(B_{t+1})|}{|W_t|}\right),&\text{(since $k_0d_0\approx k_1d_1$)}
\end{align*}
which by~\eqref{eq:W-cardinality} is what we needed to prove.
\end{proof}

\begin{lemma} For uniformly chosen $\lambda\in\widetilde{L}_t$, for every $e\in E_t$ we have $\prob_\lambda[e\in E(H^\lambda_t)]=\widetilde{O}(|E(H_t)|/|E_t|)$. %
\end{lemma}
\begin{proof}
Take any edge $e\in S_t$. By repeated applications of Constraint~\eqref{LP:d-bound}, we have
$$\sum_{\lambda\in\widetilde{L}_t:E(H^\lambda_t)\ni e}\;\sum_{K\in T^\lambda_{e}(B_t)}y_K\leq d_0^{t_0}d_1^{t_1}y_{e}.$$ As before, for any $\lambda\in\widetilde{L}_t$ such that $E(H^\lambda_t)\ni e$, the number of caterpillars this tuple of leaves contributes, $|T^\lambda_{e}(B_{t+1})|$, is roughly uniform (up to a polylogarithmic factor). Thus, we have $$|\{\lambda\in\widetilde{L}_t\mid {e}\in E(H^\lambda_t)\}||T^\lambda_{e}(B_{t+1})|y_{K_{t+1}}\approx\sum_{\lambda\in\widetilde{L}_t:E(H^\lambda_t)\ni e}\;\sum_{K\in T^\lambda_{e}(B_{t+1})}y_K\leq d_0^{t_0}d_1^{t_1}y_e.$$
By Constraints~\eqref{eq:smes1} and~\eqref{LP:d-bound} and bucketing, we have $y_e=O(k_0d_0/|E_t|)$, so the above inequality implies
$$|\{\lambda\in\widetilde{L}_t\mid E(H^\lambda_t)\ni e\}||T^\lambda_{e}(B_{t+1})|y_{K_{t+1}}=O(k_0d_0^{t_0+1}d_1^{t_1}/|E_t|).$$
Combining this with~\eqref{eq:B_t+1-sum}, we have
\begin{align*}
\prob[e\in E(H^\lambda_t)] &= \frac{|\{\lambda\in\widetilde{L}_t\mid E(H^\lambda_t)\ni e\}|}{|\widetilde{L}_t|}\\
&=\widetilde{O}\left(\frac{k_0d_0^{t_0+1}d_1^{t_1}/(|E_t||T^\lambda_{e}(B_{t+1})|y_{K_{t+1}})}{k_0d_0^{t_0+1}d_1^{t_1}/(|T^\lambda(B_{t+1})|y_{K_{t+1}})}\right)
\\&=\widetilde{O}\left(\frac{|T^\lambda(B_{t+1})|/|T^\lambda_e(B_{t+1})|}{|E_t|}\right),
\end{align*}
which by~\eqref{eq:E-cardinality} is what we needed to prove.
\end{proof}

\subsection{LP rounding: approximation guarantee}

The analysis of the approximation guarantee is similar to the combinatorial analysis of \DkS in~\cite{BCCFV10}, adapted to \SmES and slightly simplified by the regularity of degrees and LP values which we get from bucketing. %

Let us extend the previous notation by letting $S^\lambda_t$ be the set of vertices in $S_t$ which serve as rightmost endpoints of caterpillars in $T^\lambda(B_t)$. We will show that a certain invariant on %
the cardinalities and total LP values of the sets
$S^\lambda_t$ is maintained at every step $t$, as long as the required conditions for faithful rounding (i.e., the conditions for Step~\ref{step:large-degree}  or for Step~\ref{step:default} of the algorithm) do not hold.
The proof works by showing that if the invariant holds after the last iteration, then we have arrived at a contradiction. We will consider separately backbone steps and hair steps. Let us begin with backbone steps.

\begin{lemma}\label{lem:backbone} Suppose at iteration $t$ of Algorithm {\bf{Faithful-S$m$ES}} the vertices in $S_t$ are on side $b$ (where $b\in\{0,1\})$. Let $\lambda\in\widetilde{L}_t$. Then if for some $\beta\leq 1-\alpha$ the following two conditions hold
\begin{enumerate}
\item $|S^\lambda_t|y_{K_t}/y_{\lambda}=\widetilde{\Omega}(\frac{d_{1-b}}{f}\cdot f^{\beta/\alpha})$ (where $K_t$ is any caterpillar in $B_t$), and
\item $|S^\lambda_t|=\widetilde{O}(n^\beta)$,
\end{enumerate}
then either the conditions for Step~\ref{step:large-degree} hold, or the conditions for Step~\ref{step:default} hold, or the following conditions hold:
\begin{enumerate}
\item $|S^\lambda_{t+1}|y_{K_{t+1}}/y_{\lambda}\geq2\cdot\frac{d_{b}}{f}\cdot f^{(\beta+\alpha)/\alpha}$ (where $K_{t+1}$ is any caterpillar in $B_{t+1}$), and
\item $|S^\lambda_{t+1}|=\widetilde{O}(n^{\beta+\alpha})$.
\end{enumerate}
\end{lemma}

We note that the last iteration (iteration $s$) of the algorithm is always a backbone step (by our construction). We will show, in the end, that at the last step we have $\beta=1-\alpha$, and thus (assuming we do not get a faithful factor $\widetilde{O}(f)$ rounding at any point), the above lemma guarantees that $|S^\lambda_{s+1}|y_{K_{s+1}}/y_{\lambda}\geq2d_{b}\cdot f^{(1-\alpha)/\alpha}=2d_{b}(k_{1-b}f/d_b)^{1-\alpha}=2d_{b}(k_{1-b}f/d_b)/f=2k_{1-b}$. Taking $K_{s+1}$ to be the caterpillar in $B^{\lambda}_{t+1}$ with the smallest $y_{K_{s+1}}$ value, we get that $$\sum_{u\in S^{\lambda}_{t+1}}y_{\lambda\cup\{u\}}/y_\lambda\geq|S^\lambda_{s+1}|y_{K_{s+1}}/y_{\lambda}\geq2k_{1-b},$$ contradicting Constraint~\eqref{eq:smes1}. Let us now prove the above lemma.

\begin{proof}[Proof of Lemma~\ref{lem:backbone}]
By our construction, we know that for a $\widetilde{\Omega}(1)$-fraction of vertices $u\in S^{\lambda}_t$, there is some caterpillar $K_u\in T^\lambda_u(B_t)$ such that $$\sum_{w\in\Gamma(u):K_u\cup\{(u,w)\}\in B_{t+1}}y_{K_u\cup\{(u,w)\}}=\widetilde{\Omega}(d_by_{K_u}).$$
In particular, this %
 gives the following bound on the total ``relative LP value'' of edges:
\begin{equation}\label{eq:backbone-proof}
\sum_{(u,w)\in E(H^\lambda_t)}y_{K_u\cup\{(u,w)\}}/y_{\lambda} = \widetilde{\Omega}({\textstyle\frac{d_0d_1}{f}}\cdot f^{\beta/\alpha}).
\end{equation}
By Lemma~\ref{lem:LP-greedy} and Proposition~\ref{prop:final-param-ineq}, we may assume that all the degrees in $G_t=(S_t,W_t,E_t)$, and thus also in $H^{\lambda}_t$, are at most $D=\widetilde{O}(nd_0/(k_1f^2))=\widetilde{O}(n^{\alpha})$, and therefore $|S_{t+1}|\leq |S_t|\cdot D=\widetilde{O}(Cn^{\beta+\alpha})$. Thus, in this case, we only need to prove that (there exists a faithful factor $f$ rounding or) $|S^\lambda_{t+1}|y_{K_{t+1}}/y_{\lambda}\geq2\cdot\frac{d_{b}}{f}\cdot f^{(\beta+\alpha)/\alpha}$. Suppose the latter condition is not satisfied. Then by~\eqref{eq:backbone-proof}, the average degree in $H^\lambda_t$ of vertices in $W(H^\lambda_t)$ is at least
$$\widetilde{\Omega}({\textstyle\frac{d_0d_1}{f}}\cdot f^{\beta/\alpha}/({\textstyle\frac{d_{b}}{f}}\cdot f^{(\beta+\alpha)/\alpha}))=\widetilde{\Omega}(d_{1-b}/f).$$ Since, as is easy to see, the average degree of vertices in $U(H^\lambda_t)$ is $\widetilde{\Omega}(d_b)\geq\widetilde{\Omega}(d_b/f)$, by Corollary~\ref{cor:simplified-goal2} it suffices to show that $|U(H^\lambda_t)|=\widetilde{O}(k_bf)$ (and then by the corollary, there is a faithful factor $f$ rounding). However, this follows directly:
\begin{align*}
|U(H^\lambda_t)|\leq|S^\lambda_t|\leq n^{\beta}
\leq \mathrm{polylog(n)}\cdot n^{\beta}
&\leq \mathrm{polylog(n)}\cdot n^{1-\alpha}\\
&\leq \mathrm{polylog(n)}\cdot k_bf^2/d_{1-b}&\text{by Proposition~\eqref{prop:final-param-ineq}}\\
&\leq \mathrm{polylog(n)}\cdot k_b f.
\end{align*}
\end{proof}

We also have the following lemma for hair steps:

\begin{lemma}\label{lem:hair} Suppose at iteration $t$ of Algorithm {\bf{Faithful-S$m$ES}} the vertices in $S_t$ are on side $b$ (where $b\in\{0,1\})$. Let $\lambda\in\widetilde{L}_t$. Then if for some $\beta\geq 1-\alpha$ the following two conditions hold
\begin{enumerate}
\item $|S^\lambda_t|y_{K_t}/y_{\lambda}=\widetilde{\Omega}(\frac{d_{1-b}}{f}\cdot f^{\beta/\alpha})$ (where $K_t$ is any caterpillar in $B_t$), and
\item $|S^\lambda_t|=\widetilde{O}(n^\beta)$,
\end{enumerate}
then either the conditions for Step~\ref{step:large-degree} hold, or the conditions for Step~\ref{step:default} hold, or the following conditions also hold for all leaf tuples $\lambda'\in L_{t+1}$ which extend $\lambda$:
\begin{enumerate}
\item $|S^{\lambda'}_{t+1}|y_{K_{t+1}}/y_{\lambda'}=\widetilde{\Omega}(\frac{d_{1-b}}{f}\cdot f^{(\beta-(1-\alpha))/\alpha})$ (where $K_{t+1}$ is any caterpillar in $B_{t+1}$), and \label{cond:hair-proof-LP-deg}
\item $|S^\lambda_{t+1}|\leq\frac12n^{\beta-(1-\alpha)}$.
\end{enumerate}
\end{lemma}
\begin{proof}
As in the proof of Lemma~\ref{lem:backbone}, the bound on the total ``relative LP value'' of edges given by~\eqref{eq:backbone-proof} holds. By Constraint~\eqref{eq:smes1}, we also have $\sum_{w\in W(H^\lambda_t)}y_{\lambda\cup\{w\}}/y_{\lambda} \leq k_{1-b}$. Together, these inequalities give us the following lower bound on the average ``LP-degree'' of vertices $w\in W(H_t)$ (which by our construction is uniform up to a constant factor):
\begin{align}
\sum_{u:(u,w)\in E(H^\lambda_t)}y_{K_u\cup(u,w)}/y_{\lambda\cup\{w\}} &\geq {\textstyle\frac1{\mathrm{polylog(n)}}}\cdot{\textstyle\frac{d_0d_1}{fk_{1-b}}}\cdot f^{\beta/\alpha}\nonumber\\
&={\textstyle\frac1{\mathrm{polylog(n)}}}\cdot {\textstyle\frac{d_{1-b}}{f}}\cdot f^{(\beta-(1-\alpha))/\alpha}\cdot\left({\textstyle\frac{d_b}{k_{1-b}f}}\cdot f^{1/\alpha}\right)\label{eq:hair-proof}\\
&={\textstyle\frac1{\mathrm{polylog(n)}}}\cdot {\textstyle\frac{d_{1-b}}{f}}\cdot f^{(\beta-(1-\alpha))/\alpha}.\nonumber
\end{align}
Thus, condition~\ref{cond:hair-proof-LP-deg} always holds. Thus, we need to show that, if there is no faithful factor $f$ rounding, the average degree (in fact maximum degree, but these are only a constant factor apart) in $H^{\lambda}_t$ of vertices $w\in W(H^\lambda_t)$ is at most $\frac12n^{\beta-(1-\alpha)}$. Suppose not. That is, suppose the average degree of such vertices is $\Omega(n^{\beta-(1-\alpha)})$, and let us show that we can get a faithful rounding of factor $f$.

Note that by~\eqref{eq:hair-proof}, the average degree is at least $\widetilde{\Omega}(\frac{d_{1-b}}{f}\cdot f^{(\beta-(1-\alpha))/\alpha})\geq\widetilde{\Omega}(\frac{d_{1-b}}{f})$. Since, as before, it is easy to see that the average degree of vertices $u\in U(H^\lambda_t)$ is at least $\widetilde{\Omega}(d_b)\geq \widetilde{\Omega}(d_b/f)$, by Corollary~\ref{cor:simplified-goal1} we only need to show that $d'_{1-b}/|U(H^\lambda_t)| \geq d_0/(k_1f^2)$, where $d'_{1-b}$ denotes the average degree of vertices $w\in W(H^\lambda_t)$. Indeed, by our assumption about this average degree, we have
\begin{align*}
\frac{d'_{1-b}}{|U(H^\lambda_t)|}\geq\Omega\left(\frac{n^{\beta-(1-\alpha)}}{|U(H^\lambda_t)|}\right)\geq\Omega\left(\frac{n^{\beta-(1-\alpha)}}{|S^\lambda_t|}\right)=\widetilde{\Omega}\left(n^{-(1-\alpha)}\right)\geq\frac{d_0}{k_1f^2},
\end{align*}
where the last inequality follows from Proposition~\eqref{prop:final-param-ineq}.
\end{proof}

Combining these two lemmas, the main theorem easily follows:

\begin{theorem} (Assuming $f\leq d_0$,) Algorithm {\bf{Faithful-S$m$ES}} gives either a factor $f$ faithful rounding for \SmES, or it gives a skewed proportional rounding satisfying the conditions of Lemma~\ref{lem:simplified-goal}.
\end{theorem}
\begin{proof}
Denote by $b_t$ the side which contains the vertices in $S_t$. The theorem follows from the following simple claim, which can be proved directly by induction using the above lemmas (we use the notation $\{x\}=x+1-\lceil x\rceil$):
\begin{claim} Suppose the conditions for Step~\ref{step:large-degree} never hold. Then for all $t=2,\ldots,s+1$, either the conditions for Step~\ref{step:default} hold at at least one of the iterations through step $t-1$, or we have
\begin{enumerate}
\item $|S^\lambda_t|y_{K_t}/y_{\lambda}=\widetilde{\Omega}(\frac{d_{1-{b_t}}}{f}\cdot f^{\{(t-1)\alpha\}/\alpha})$ (for all $\lambda\in L_t$, and caterpillar $K_t\in T^\lambda(B_t)$), and
\item $|S^\lambda_t|=\widetilde{O}(n^{\{(t-1)\alpha\}})$.
\end{enumerate}
Moreover, if $\{(t-1)\alpha\}\geq\alpha$ (i.e.\ step $t-1$ was a backbone step), then $$\sum_{u\in S^\lambda_t}y_{\lambda\cup\{u\}}/y_\lambda\geq2\cdot \frac{d_{1-{b_t}}}{f}\cdot f^{\{(t-1)\alpha\}/\alpha}.$$
\end{claim}
As pointed out earlier, this last inequality yields a contradiction for $t=s+1$. Therefore one of the steps gives a factor-$f$ faithful rounding.
\end{proof}

\section{Discussion and Future Directions} \label{sec:discussion}

Some features of our techniques might be applicable to other problems.
Most obviously, this is perhaps the first time that LP hierarchies are applied
to ``local'' parts of an LP, rather than to the entire LP.
Can this approach be useful for other problems?
Currently, it is not clear to us how this approach fares against
one ``global'' application of an LP hierarchy to some basic relaxation:
 a global hierarchy could take advantage of non-locality in the constraints and solution, but on the other hand would not allow us to locally ``guess'' degrees (see e.g.~footnote~\ref{foot:guess}).

Persistent gaps in the approximability of other network design problems
naturally call for a judicious use of LP hierarchies in order
to obtain better approximation algorithms.
For example, the {\sc basic $k$-spanner} problem,
in which the goal is to construct a $k$-spanner with as few edges as possible,
is only known to admit approximation ratio $O(n^{\lceil 2/(k+1) \rceil})$~\cite{ADDJS93}, while the best hardness of approximation is $2^{(\log^{1-\eps} n) /k}$ for arbitrarily small constant $\eps > 0$~\cite{DKR12}.
An integrality gap that almost matches the upper bound (namely a gap of $n^{\Omega(1/k)}$) was recently shown by Dinitz and Krauthgamer~\cite{DK11a},
but stronger relaxations obtained via hierarchies
can possibly have smaller integrality gaps.  In particular, it is not at all clear what the best achievable approximation ratio is for the regime when $k$ is constant; perhaps hierarchies will finally allow upper bounds that beat~\cite{ADDJS93}.
Similarly, for directed $k$-spanner the known upper bound
is $\tO(\sqrt{n})$ \cite{BBMRY11},
and there is an $\tOmega(n^{1/3})$ integrality gap~\cite{DK11a},
but it only applies to a simple LP relaxation.
Yet other relevant problems are {\sc Directed Steiner Tree}
and {\sc Directed Steiner Forest}, see \cite{MKN09,BBMRY11} and references therein.
Perhaps hierarchies could help for any of these problems?

Finally, the connection we show between \LD2S and \SmES
suggests an intriguing possibility for conditional lower bounds.
The current hardness for \LD2S is only $\Omega(\log n)$,
while \SmES is basically as hard as \DkS, which is commonly thought to be
difficult to approximate well (say within a polylogarithmic factor, although current hardness results rely on various complexity assumptions
and give only a relatively small constant \cite{F02,Khot06}).
A reduction in the other direction, i.e.~from \SmES to \LD2S,
could give partial evidence that \LD2S cannot be approximated well,
and could possibly even match the upper bound that we prove here.
The same arguments about a formal connection to \DkS
obviously apply also to other network design problems,
such as {\sc basic $k$-spanner}.

{\small
\bibliographystyle{alphainit}
\bibliography{spanner,robi,drafts}

\newcommand{\etalchar}[1]{$^{#1}$}
\begin{thebibliography}{BBM{\etalchar{+}}11}

\bibitem[ABL02]{ABL02}
S.~Arora, B.~Bollob{\'a}s, and L.~Lov{\'a}sz.
\newblock Proving integrality gaps without knowing the linear program.
\newblock In {\em 43rd Annual IEEE Symposium on Foundations of Computer
  Science}, pages 313--322, 2002.

\bibitem[ABP92]{ABP92}
B.~Awerbuch, A.~Baratz, and D.~Peleg.
\newblock Efficient broadcast and lightweight spanners.
\newblock Technical Report CS92-22, Weizmann Institute of Science, 1992.

\bibitem[ADD{\etalchar{+}}93]{ADDJS93}
I.~Alth\"{o}fer, G.~Das, D.~Dobkin, D.~Joseph, and J.~Soares.
\newblock On sparse spanners of weighted graphs.
\newblock {\em Discrete Comput. Geom.}, 9(1):81--100, 1993.

\bibitem[AGGN10]{AGGN10}
B.~Anthony, V.~Goyal, A.~Gupta, and V.~Nagarajan.
\newblock A plant location guide for the unsure: Approximation algorithms for
  min-max location problems.
\newblock {\em Math. Oper. Res.}, 35(1):79--101, February 2010.

\bibitem[AP95]{AP95}
B.~Awerbuch and D.~Peleg.
\newblock Online tracking of mobile users.
\newblock {\em J. ACM}, 42(5):1021--1058, 1995.

\bibitem[BBM{\etalchar{+}}11]{BBMRY11}
P.~Berman, A.~Bhattacharyya, K.~Makarychev, S.~Raskhodnikova, and
  G.~Yaroslavtsev.
\newblock Improved approximation for the directed spanner problem.
\newblock In {\em 38th International Colloquium on Automata, Languages and
  Programming}, volume 6755 of {\em Lecture Notes in Computer Science}, pages
  1--12. Springer, 2011.

\bibitem[BCC{\etalchar{+}}10]{BCCFV10}
A.~Bhaskara, M.~Charikar, E.~Chlamtac, U.~Feige, and A.~Vijayaraghavan.
\newblock Detecting high log-densities: an {$O(n^{1/4})$} approximation for
  densest $k$-subgraph.
\newblock In {\em 42nd ACM Symposium on Theory of Computing}, pages 201--210,
  2010.

\bibitem[BCG09]{BCG09}
M.~Bateni, M.~Charikar, and V.~Guruswami.
\newblock Maxmin allocation via degree lower-bounded arborescences.
\newblock In {\em 41st annual ACM symposium on Theory of computing}, pages
  543--552, New York, NY, USA, 2009. ACM.

\bibitem[BGJ{\etalchar{+}}09]{BGJRW09}
A.~Bhattacharyya, E.~Grigorescu, K.~Jung, S.~Raskhodnikova, and D.~P. Woodruff.
\newblock Transitive-closure spanners.
\newblock In {\em 20th Annual ACM-SIAM Symposium on Discrete Algorithms}, pages
  932--941, 2009.

\bibitem[BRR10]{BRR10}
P.~Berman, S.~Raskhodnikova, and G.~Ruan.
\newblock {Finding Sparser Directed Spanners}.
\newblock In {\em FSTTCS 2010}, volume~8 of {\em Leibniz International
  Proceedings in Informatics (LIPIcs)}, pages 424--435, 2010.

\bibitem[BRS11]{BRS11}
B.~Barak, P.~Raghavendra, and D.~Steurer.
\newblock Rounding semidefinite programming hierarchies via global correlation.
\newblock In {\em 52nd Annual IEEE Symposium on Foundations of Computer Science
  (FOCS'11)}. IEEE, 2011.

\bibitem[Chl07]{Chlamtac07}
E.~Chlamtac.
\newblock Approximation algorithms using hierarchies of semidefinite
  programming relaxations.
\newblock In {\em 48th Annual IEEE Symposium on Foundations of Computer
  Science}, pages 691--701. IEEE Computer Society, 2007.

\bibitem[CKR10]{CKR10}
E.~Chlamtac, R.~Krauthgamer, and P.~Raghavendra.
\newblock Approximating sparsest cut in graphs of bounded treewidth.
\newblock In {\em 13th International Workshop on Approximation, Randomization,
  and Combinatorial Optimization}, volume 6302 of {\em Lecture Notes in
  Computer Science}, pages 124--137. Springer, 2010.

\bibitem[CS08]{CS08}
E.~Chlamtac and G.~Singh.
\newblock Improved approximation guarantees through higher levels of {SDP}
  hierarchies.
\newblock In {\em 11th International Workshop on Approximation Algorithms for
  Combinatorial Optimization Problems (APPROX)}, pages 49--62. Springer-Verlag,
  2008.

\bibitem[CT12]{CT12}
E.~Chlamtac and M.~Tulsiani.
\newblock Convex relaxations and integrality gaps.
\newblock In M.~F. Anjos and J.~B. Lasserre, editors, {\em Handbook on
  Semidefinite, Conic and Polynomial Optimization}, volume 166 of {\em
  International Series in Operations Research \& Management Science}, pages
  139--169. Springer, 2012.

\bibitem[Din07]{Din07}
M.~Dinitz.
\newblock Compact routing with slack.
\newblock In {\em 26th annual ACM Symposium on Principles of Distributed
  Computing}, pages 81--88. ACM, 2007.

\bibitem[DK11a]{DK11a}
M.~Dinitz and R.~Krauthgamer.
\newblock Directed spanners via flow-based linear programs.
\newblock In {\em 43rd Annual ACM Symposium on Theory of Computing}, pages
  323--332. ACM, 2011.

\bibitem[DK11b]{DK11b}
M.~Dinitz and R.~Krauthgamer.
\newblock Fault-tolerant spanners: better and simpler.
\newblock In {\em 30th Annual ACM Symposium on Principles of Distributed
  Computing}, pages 169--178. ACM, 2011.

\bibitem[DKR12]{DKR12}
M.~Dinitz, G.~Kortsarz, and R.~Raz.
\newblock Label cover instances with large girth and the hardness of
  approximating basic k-spanner.
\newblock Available at http://arxiv.org/abs/1203.0224, 2012.

\bibitem[EEST08]{EEST08}
M.~Elkin, Y.~Emek, D.~A. Spielman, and S.-H. Teng.
\newblock Lower-stretch spanning trees.
\newblock {\em SIAM J. Comput.}, 38(2):608--628, 2008.

\bibitem[EP01]{EP01}
M.~Elkin and D.~Peleg.
\newblock The client-server 2-spanner problem with applications to network
  design.
\newblock In {\em 8th International Colloquium on Structural Information and
  Communication Complexity (SIROCCO)}, pages 117--132, 2001.

\bibitem[Fei02]{F02}
U.~Feige.
\newblock Relations between average case complexity and approximation
  complexity.
\newblock In {\em Proceedings of the 34th annual ACM Symposium on Theory of
  Computing (STOC'02)}, pages 534--543. ACM Press, 2002.

\bibitem[FKN09]{MKN09}
M.~Feldman, G.~Kortsarz, and Z.~Nutov.
\newblock Improved approximating algorithms for directed {S}teiner forest.
\newblock In {\em 20th Annual ACM-SIAM Symposium on Discrete Algorithms}, pages
  922--931. SIAM, 2009.

\bibitem[FS97]{FS97}
U.~Feige and M.~Seltser.
\newblock On the densest $k$-subgraph problem.
\newblock Technical report, Weizmann Institute of Science, Rehovot, Israel,
  1997.

\bibitem[GHNR07]{GHNR07}
A.~Gupta, M.~Hajiaghayi, V.~Nagarajan, and R.~Ravi.
\newblock Dial a ride from k-forest.
\newblock In L.~Arge, M.~Hoffmann, and E.~Welzl, editors, {\em Algorithms – ESA
  2007}, volume 4698 of {\em Lecture Notes in Computer Science}, pages
  241--252. Springer Berlin / Heidelberg, 2007.

\bibitem[GS11]{GS11}
V.~Guruswami and A.~K. Sinop.
\newblock Lasserre hierarchy, higher eigenvalues, and approximation schemes for
  graph partitioning and quadratic integer programming with {PSD} objectives.
\newblock In {\em 52nd Annual IEEE Symposium on Foundations of Computer Science
  (FOCS'11)}. IEEE, 2011.

\bibitem[Kho06]{Khot06}
S.~Khot.
\newblock Ruling out {PTAS} for graph min-bisection, dense {$k$}-subgraph, and
  bipartite clique.
\newblock {\em SIAM J. Comput.}, 36(4):1025--1071, 2006.

\bibitem[Kor01]{Kortsarz01}
G.~Kortsarz.
\newblock On the hardness of approximating spanners.
\newblock {\em Algorithmica}, 30(3):432--450, 2001.

\bibitem[KP93]{KP93}
G.~Kortsarz and D.~Peleg.
\newblock On choosing a dense subgraph.
\newblock In {\em 34th Annual Symposium on Foundations of Computer Science},
  pages 692 --701, 1993.

\bibitem[KP94]{KP94}
G.~Kortsarz and D.~Peleg.
\newblock Generating sparse 2-spanners.
\newblock {\em J. Algorithms}, 17(2):222--236, 1994.

\bibitem[KP98]{KP98}
G.~Kortsarz and D.~Peleg.
\newblock Generating low-degree 2-spanners.
\newblock {\em SIAM J. Comput.}, 27(5):1438--1456, 1998.

\bibitem[Las02]{Lasserre02}
J.~B. Lasserre.
\newblock An explicit equivalent positive semidefinite program for nonlinear
  0-1 programs.
\newblock {\em SIAM Journal on Optimization}, 12(3):756--769, 2002.

\bibitem[LNSS09]{LNSS09}
L.~C. Lau, J.~S. Naor, M.~R. Salavatipour, and M.~Singh.
\newblock Survivable network design with degree or order constraints.
\newblock {\em SIAM J. Comput.}, 39(3):1062--1087, 2009.

\bibitem[LS91]{LS91}
L.~Lov\'asz and A.~Schrijver.
\newblock Cones of matrices and set-functions and 0-1 optimization.
\newblock {\em SIAM Journal on Optimization}, 1(12):166--190, 1991.

\bibitem[Nut10]{Nut10}
Z.~Nutov.
\newblock Approximating steiner networks with node-weights.
\newblock {\em SIAM Journal on Computing}, 39(7):3001--3022, 2010.

\bibitem[PS89]{PS89}
D.~Peleg and A.~A. Sch{\"a}ffer.
\newblock Graph spanners.
\newblock {\em J. Graph Theory}, 13(1):99--116, 1989.

\bibitem[PU89]{PU89}
D.~Peleg and J.~D. Ullman.
\newblock An optimal synchronizer for the hypercube.
\newblock {\em SIAM J. Comput.}, 18:740--747, August 1989.

\bibitem[SA90]{SA90}
H.~D. Sherali and W.~P. Adams.
\newblock A hierarchy of relaxation between the continuous and convex hull
  representations.
\newblock {\em SIAM J. Discret. Math.}, 3(3):411--430, 1990.

\bibitem[SL07]{SL07}
M.~Singh and L.~C. Lau.
\newblock Approximating minimum bounded degree spanning trees to within one of
  optimal.
\newblock In {\em STOC '07: Proceedings of the thirty-ninth annual ACM
  Symposium on Theory of Computing}, pages 661--670, New York, NY, USA, 2007.
  ACM.

\bibitem[ST04]{ST04a}
D.~A. Spielman and S.-H. Teng.
\newblock Nearly-linear time algorithms for graph partitioning, graph
  sparsification, and solving linear systems.
\newblock In {\em 36th Annual ACM Symposium on Theory of Computing}, pages
  81--90. ACM, 2004.

\bibitem[TZ01]{TZ01r}
M.~Thorup and U.~Zwick.
\newblock Compact routing schemes.
\newblock In {\em 13th annual ACM Symposium on Parallel Algorithms and
  Architectures}, pages 1--10. ACM Press, 2001.

\bibitem[TZ05]{TZ05}
M.~Thorup and U.~Zwick.
\newblock Approximate distance oracles.
\newblock {\em J. ACM}, 52(1):1--24, 2005.

\end{thebibliography}
}

\appendix

\section{Handling small degrees}\label{sec:small-degrees}
Just as in~\cite{BCCFV10}, we may assume that the subgraph degrees (in the optimum, or according to the LP) are greater than or equal to our desired approximation ratio $f$. The reason is that there is always a simple $d_0$-approximation (recall our convention that $d_0\leq d_1$). Thus if $d_0\leq f$, we are done.

\begin{lemma}\label{lem:small-deg-algo}
There exists a faithful rounding with factor $\tO(d_0)$.
\end{lemma}
\begin{proof}
Consider the following algorithm. Divide all edges in the graph into buckets by their LP-values, and also by the LP values of their endpoint vertices (so that all three parameters are uniform up to a constant factor within each bucket). Then there is some bucket $B$ for which $\sum_{e\in B}y_e=\widetilde{\Omega}(m)$. For $b=0,1$, let $U_b$ be the set of vertices on side $b$ which have at least one edge in $B$ incident to them. Now pick a subset $U'_1\subseteq U_1$ of size $k_1$ uniformly at random, and let $B(U'_1)$ be the set of edges in $B$ incident to vertices in $U'_1$. Finally, choose $B'\subseteq B(U'_1)$ of size 
\begin{equation}\label{eq:small-degree-proof1}
|B(U'_1)|\cdot\frac{m}{|B|}\cdot\frac{|U_1|}{k_1}
\end{equation}
uniformly at random. We return the graph induced by the edges in $B'$.

Note that by uniformity of LP values in $B$, every edge $e\in B$ has LP weight $y_e=\tilde{\Theta}(m/|B|)$. Moreover, since by Constraint~\eqref{LP:d-bound}, vertices on side $b$ (for $b=0,1$) all contribute LP degree (that is, the LP weight of its incident edges divided by its LP value) $\tilde{\Theta}(d_b)$, then the LP weight of every vertex $u\in U_b$ is $\tilde{\Theta}(k_b/|U_b|)$. For brevity, let us omit polylogarithmic factors in the remaining discussion. Note that since for any vertex $u$ and edge $e$ incident to $u$ we have $y_e\leq y_u$, this implies
\begin{equation}\label{eq:small-degree-proof2}
\frac{m}{|B|}\leq\frac{k_1}{|U_1|},
\end{equation}
which also shows that the quantity in~\eqref{eq:small-degree-proof1} is at most $|B(U'_1)|$.

For vertices $u_1\in U_1$, the probability that $u_1\in U'_1$ is $k_1/|U_1|=y_{u_1}$. This, together with the fact that $|U'_1|=k_1$ gives faithfulness for vertices in $U_1$ even for approximation factor $1$. Moreover, for any edge $e\in B$, the probability that $e\in B(U'_1)$ is also $k_1/|U_1|$. Conditioned on $e\in B$, the probability that $e$ is retained in $B'$ is $\frac{m}{|B|}\cdot\frac{|U_1|}{k_1}$, and thus the probability (a priori) that we will have $e\in B'$ is $m/|B|=y_e$, which gives faithfulness for edges.

Next, we would like to show that at most $m$ edges are chosen, which will also give the required bound on the number of vertices in $U_0$ chosen, since $m=d_0k_0$. Note that since the LP degree of every vertex $u_1\in U_1$ is at most $d_1$, by the LP values of edges and vertices in $B$, this implies that the graph degree of every vertex in $U_b$ (for $b=0,1$) is at most 
\begin{equation}\label{eq:small-degree-proof3}
D_b= d_b\cdot\frac{k_b/|U_b|}{m/|B|}.
\end{equation}
 Thus, the total number of edges in $B(U'_1)$ is always at most $|U'_1|\cdot d_1k_1|B|/(|U_1|m)$. By~\eqref{eq:small-degree-proof1}, the number of edges in $B'$ is then always at most $|U'_1|\cdot d_1=k_1d_1=m$.

It remains to analyze the probability that an individual vertex in $U_0$ is chosen. Since every vertex $u_0\in U_0$ is picked iff one of its incident edges is included in $B'$, and since the probability of each such event is at most $y_e$, by a union bound, and by~\eqref{eq:small-degree-proof3}, the probability that $u_0$ is picked is at most
$$D_0y_e=d_0\cdot\frac{k_0/|U_0|}{m/|B|}\cdot y_e=d_0\cdot\frac{k_0/|U_0|}{m/|B|}\cdot \frac{m}{|B|}=d_0\cdot\frac{k_0}{|U_0|}=d_0\cdot y_{u_0}.$$
\end{proof}

\section{Bounding the maximum degree}\label{sec:max-deg-bound}
It remains to show the correctness of Step~\ref{step:large-degree} in Algorithm {\bf{Faithful-S$m$ES}}. Before we do this, let us start with a simple claim.

\begin{claim}\label{claim:ave-is-max} In the graph $G_t$, the average degree and maximum degree of vertices in $W_t$ differ by at most a polylogarithmic factor (and the same holds for vertices in $S_t$).
\end{claim}
\begin{proof}
Let us show this for $W_t$ (the proof for $S_t$ is identical). By Lemma~\ref{lem:regularity}, the LP degree of the average vertex $w\in W_t$ (that is, the quantity $\sum_{u:(u,w)\in E_t}y_{\{u,w\}}/y_{\{w\}}$) is $\Omega(d_1)$, and by Constraint~\eqref{LP:d-bound} no LP-degree (for vertices in $W_t$) can be more than $d_1$. Moreover, by uniformity of LP values, the LP-degrees are proportional to the graph degrees, which proves the claim.
\end{proof}

\begin{lemma}\label{lem:LP-greedy} If at any iteration $t$ Algorithm {\bf{Faithful-S$m$ES}}, the maximum degree in $G_t$ is at least $\widetilde\Omega\left(\frac{nd_0}{k_1f^2}\right)$, then Step~\ref{step:large-degree} gives a faithful factor-$f$ rounding.
\end{lemma}
\begin{proof}
Without loss of generality, assume $S_t\subseteq V_0$ (the case where $S_t\subseteq V_1$ is essentially the same). By Lemma~\ref{lem:regularity} and uniformity of LP values (which we get from bucketing), we have that for every vertex $u\in S_t$, $y_{\{u\}}\approx k_0/|S_t|$, for every vertex $w\in W_t$, $y_{\{w\}}\approx k_1/|W_t|$, and for every edge $e\in E_t$, $y_{\{e\}}\approx m/|E_t|$.

Thus, it is clear that the random sampling of vertices in Step~\ref{step:large-degree} is faithful with respect to vertices. It suffices, then,  to show that every edge $e\in E_t$ is picked with probability $\widetilde{\Theta}(y_{\{e\}})$ (this is sufficient since they have total LP weight $\widetilde{\Omega}(m)$). Indeed, since every vertex $v\in S_t\cup W_t$ is chosen with probability $\widetilde{\Theta}(y_{\{v\}} f)$, the probability that both endpoints of an edge in $e=(u,w)\in E_t$ are chosen is $\widetilde{\Theta}(y_{\{u\}}y_{\{w\}}f^2)$. Hence we only need to show that this probability is at least $y_{\{e\}}$. This follows since, by the maximum degree bound, and by Claim~\ref{claim:ave-is-max}, we have that (ignoring polylogarithmic factors):
$$|E_t|\geq |W_t|\cdot\frac{nd_0}{k_1f^2}\approx\frac{nd_0}{y_{\{w\}}f^2}\approx\frac{nm}{y_{\{w\}}f^2k_0}\geq\frac{|S_t|m}{y_{\{w\}}f^2k_0}\approx\frac{m}{y_{\{u\}}y_{\{w\}}f^2}\approx\frac{y_{\{e\}}|E_t|}{y_{\{u\}}y_{\{w\}}f^2}.$$
\end{proof}

\end{document}